\documentclass[a4paper,11pt]{article}
\usepackage[margin=1in]{geometry}
\usepackage[T1]{fontenc}

\usepackage{graphicx}

\usepackage[numbers,sort&compress]{natbib}

\usepackage{amsmath,amssymb,amsthm}

\newtheorem{theorem}{Theorem}
\newtheorem{lemma}{Lemma}
\newtheorem{corollary}{Corollary}
\let\subsubsection\paragraph

\usepackage{dsfont}
\usepackage{comment,hyperref}
\usepackage{tikz}
\usepackage{tikz-network}
\usepackage{pgfplots}
\usepackage{varwidth}
\usepackage{enumitem}
\pgfplotsset{compat=1.10}
\usepgfplotslibrary{fillbetween}
\usetikzlibrary{backgrounds}
\usetikzlibrary{patterns}
\usetikzlibrary{arrows}
\usepackage{multirow}
\usepackage[ruled]{algorithm2e}
\usepackage{xcolor}

\newlength{\algofontsize}
\hypersetup{
	colorlinks,
	linkcolor={red!50!black},
	citecolor={blue!50!black},
	urlcolor={blue!80!black}
}

\newcommand{\calG}{\mathcal{G}}
\newcommand{\calN}{\mathcal{N}}
\newcommand{\calP}{\mathcal{P}}
\newcommand{\RR}{\mathbb{R}}
\newcommand{\NN}{\mathbb{N}}
\newcommand{\ie}{i.e.,\xspace}

\DeclareMathOperator{\ttop}{top}

\def\equationautorefname~#1\null{(#1)\null}

\begin{document}
	%
	\title{Optimal Impartial Correspondences}
	
	\author{%
		Javier Cembrano\thanks{Institut für Mathematik, Technische Universität Berlin, Germany}
		\and
		Felix Fischer\thanks{School of Mathematical Sciences, Queen Mary University of London, UK}
		\and
		Max Klimm\thanks{Institut für Mathematik, Technische Universität Berlin, Germany}
	}
	\date{\vspace{-1em}}
	
	\maketitle              
	\begin{abstract}
		We study mechanisms that select a subset of the vertex set of a directed graph in order to maximize the minimum indegree of any selected vertex, subject to an impartiality constraint that the selection of a particular vertex is independent of the outgoing edges of that vertex. For graphs with maximum outdegree $d$, we give a mechanism that selects at most $d+1$ vertices and only selects vertices whose indegree is at least the maximum indegree in the graph minus one. We then show that this is best possible in the sense that no impartial mechanism can only select vertices with maximum degree, even without any restriction on the number of selected vertices. We finally obtain the following trade-off between the maximum number of vertices selected and the minimum indegree of any selected vertex: when selecting at most~$k$ vertices out of $n$, it is possible to only select vertices whose indegree is at least the maximum indegree minus $\lfloor(n-2)/(k-1)\rfloor+1$.
	\end{abstract}

	\section{Introduction}
	
	Impartial selection is the problem of selecting vertices with large indegree in a directed graph, in such a way that the selection of a particular vertex is independent of the outgoing edges of that vertex. The problem models a situation where agents nominate one another for selection and are willing to offer their true opinion on other agents as long as this does not affect their own chance of being selected. 
	
	The selection of a single vertex is governed by strong impossibility results. For graphs with maximum outdegree one, corresponding to situations where each agent submits a single nomination, every impartial selection rule violates one of two basic axioms~\citep{holzman2013impartial} and as a consequence must fail to provide a non-trivial multiplicative approximation to the maximum indegree. For graphs with arbitrary outdegrees, corresponding to situations where each agent can submit multiple nominations, impartial rules violate an even weaker axiom and cannot provide a non-trivial approximation in a multiplicative or additive sense~\citep{alon2011sum,cembrano2022impartial}. These impossibilities largely remain in place if rather than a single vertex we want to select any fixed number of vertices, but positive results can be obtained if we relax the requirement that the same number of vertices must be selected in every graph~\citep{TaOh14a,bjelde2017impartial}.
	
	From a practical point of view, the need for such a relaxation should not necessarily be a cause for concern. Indeed, situations in the real world to which impartial selection is relevant often allow for a certain degree of flexibility in the number of selected agents. The exact number of papers accepted to an academic conference is usually not fixed in advance but depends on the number and quality of submissions. Best paper awards at conferences are often given in overlapping categories, and some awards may only be given if this is warranted by the field of candidates. The Fields medal is awarded every four years to two, three, or four mathematicians under the age of $40$. Examples at the more extreme end of the spectrum of flexibility include the award of job titles such as vice president or deputy vice-principal. Such titles can often be given to a large number of individuals at a negligible cost per individual, but should only be given to qualified individuals so as not to devalue the title.
	
	\citet{TaOh14a} specifically studied what they call nomination correspondences, \ie rules that may select an arbitrary set of vertices in any graph. For graphs with maximum outdegree one a particular such rule, plurality with runners-up, satisfies impartiality and appropriate versions of the two axioms of \citet{holzman2013impartial}. The rule selects any vertex with maximum indegree; if there is a unique such vertex, any vertex whose indegree is smaller by one and whose outgoing edge goes to the vertex with maximum indegree is selected as well. An appropriate measure for the quality of rules that select varying numbers of vertices is the difference in the worst case between the best vertex and the worst selected vertex, and we can call a rule $\alpha$-min-additive if the maximum difference, taken over all graphs, between these two quantities is at most $\alpha$. In this terminology, plurality with runners-up is $1$-min-additive.
	
	As \citeauthor{TaOh14a} point out, it may be desirable in practice to ensure that the maximum number of vertices selected is not too large, a property that plurality with runners-up clearly fails. It is therefore interesting to ask whether there exist rules that are $\alpha$-min-additive and never select more than $k$ vertices, for some fixed $\alpha$ and $k$. For graphs with outdegree one, \citeauthor{TaOh14a} answer this question in the affirmative: a variant of plurality with runners-up that breaks ties according to a fixed ordering of the vertices remains $1$-min-additive but never selects more than two vertices.
	
	\subsubsection{Our Contribution}
	
	Our first result provides a generalization of the result of \citeauthor{TaOh14a} to graphs with larger outdegrees: for graphs with maximum outdegree $d$, it is possible to achieve $1$-min-additivity while selecting at most $d+1$ vertices. For the particular case of graphs with unbounded outdegrees we obtain a slight improvement, by guaranteeing $1$-min-additivity without ever selecting all vertices. 
	Our second result establishes that $1$-min-additivity is best possible, thus ruling out the existence of impartial mechanisms that only select vertices with maximum indegree. This holds even when no restrictions are imposed on the number of selected vertices, and is shown alongside analogous impossibility results concerning the maximization of the median or mean indegree of the selected vertices instead of their minimum indegree. 
	Our third result provides a trade-off between the maximum number of vertices selected, where smaller is better, and the minimum indegree of any selected vertex, where larger is better: if we are allowed to select at most $k$ vertices out of $n$, we can guarantee $\alpha$-min-additivity for $\alpha=\lfloor (n-2)/(k-1)\rfloor+1$. This is achieved by removing a subset of the edges from the graph before plurality with runners-up is applied, in order to guarantee impartiality while selecting fewer vertices. We do not know whether this last result is tight and leave open the interesting question for the optimal trade-off between the number and quality of selected vertices.
	
	\subsubsection{Related Work}
	
	Impartiality as a property of an economic mechanism was introduced by \citet{de2008impartial}, and first applied to the selection of vertices in a directed graph by \citet{alon2011sum} and \citet{holzman2013impartial}. Whereas \citeauthor{holzman2013impartial} gave axiomatic characterizations for mechanisms selecting a single vertex when all outdegrees are equal to one, \citeauthor{alon2011sum} studied the ability of impartial mechanisms to approximate the maximum indegree for any fixed number of vertices when there are no limitations on outdegrees.
	
	Both sets of authors obtained strong impossibility results, which a significant amount of follow-up work has since sought to overcome. Randomized mechanisms providing non-trivial multiplicative guarantees had already been proposed by \citeauthor{alon2011sum}, and \citet{fischer2015optimal} subsequently achieved the best possible such guarantee for the selection of one vertex. Starting from the observation that worst-case instances for randomized mechanisms have small indegrees, \citet{bousquet2014near} developed a mechanism that is asymptotically optimal as the maximum indegree grows, and \citet{caragiannis2019impartial, caragiannis2021impartial} initiated the study of mechanisms providing additive rather than multiplicative guarantees. \citet{cembrano2022impartial} subsequently identified a deterministic mechanism that provides non-trivial additive guarantees whenever the maximum outdegree is bounded and established that no such guarantees can be obtained with unbounded outdegrees. Randomized mechanisms have been also studied from an axiomatic point of view by \citet{mackenzie2015symmetry, MacK20a}.
	
	\citet{bjelde2017impartial} gave randomized mechanisms with improved multiplicative guarantees for the selection of more than one vertex and observed that when selecting at most $k$ vertices rather than exactly $k$, deterministic mechanisms can in fact achieve non-trivial guarantees. An axiomatic study of \citet{TaOh14a} for the outdegree-one case came to the same conclusion: when allowing for the selection of a varying number of vertices, the impossibility result of \citeauthor{holzman2013impartial} no longer holds. \citet{Tamu16a} subsequently characterized a mechanism proposed by \citeauthor{TaOh14a}, which in some cases selects all vertices, as the unique minimal mechanism satisfying impartiality, anonymity, symmetry, and monotonicity.
	
	Impartial mechanisms have finally been proposed for various problems other than selection, including peer review \citep{aziz2019strategyproof, kurokawa2015impartial, mattei2020peernomination, xu2018strategyproof}, rank aggregation \citep{kahng2018ranking}, progeny maximization \citep{babichenko2020incentive, zhang2021incentive}, and network centralities \citep{wkas2019random}.

	\section{Preliminaries}
	
	For $n\in \NN$, let $[n]=\{1,2,\dots,n\}$, and let 
	\[
	\calG_n = \left\{(V,E):V=[n],E \subseteq (V\times V) \setminus \bigcup_{v\in V}\{(v,v)\}\right\}
	\]
	be the set of directed graphs with $n$ vertices and no loops. Let $\calG = \bigcup_{n\in \NN} \calG_n$. For $G=(V,E)\in\calG$ and $v\in V$, let $N^+(v, G)=\{u\in V: (v,u) \in E\}$ be the out-neighborhood and $N^-(v, G)=\{u\in V:(u,v)\in E\}$ the in-neighborhood of~$v$ in~$G$. Let $\delta^+(v,G)=|N^+(v,G)|$ and $\delta^-(v,G)=|N^-(v,G)|$ denote the outdegree and indegree of~$v$ in~$G$,
	and $\Delta(G) = \max_{v\in V} \delta^-(v, G)$ the maximum indegree of any vertex in $G$.
	When the graph is clear from the context, we will sometimes drop $G$ from the notation and write $N^+(v)$, $N^-(v)$, $\delta^+(v)$, $\delta^-(v)$, 
	and~$\Delta$.
	Let $\ttop(G)=\max\{v\in V:\delta^-(v)=\Delta(G)\}$ denote the vertex of $G$ with the largest index among those with maximum indegree.
	For $n\in\NN$ and $d\in[n-1]$, let $\calG_n(d) = \{(V, E)\in \calG_n: \text{$\delta^+(v)\leq d$ for every $v\in V$}\}$ be the set of graphs in $\calG_n$ with maximum outdegree at most $d$, and $\calG(d) = \bigcup_{n\in \NN} \calG_n(d)$.
	
	A \textit{$k$-selection mechanism} is then given by a family of functions $f:\calG_n \to 2^{[n]}$, one for each $n\in\NN$, mapping each graph to a subset of its vertices, where we require that $|f(G)|\leq k$ for all $G\in\calG$.
	In a slight abuse of notation, we will use $f$ to refer to both the mechanism and to individual functions of the family.
	Given $G=(V,E)\in \calG$ and $v\in V$, let $\calN_v(G)=\{(V,E')\in \calG:\ E\setminus(\{v\}\times V)=E'\setminus(\{v\}\times V)\}$ be the set neighboring graphs of $G$ with respect to $v$, in the sense that they can be obtained from $G$ by changing the outgoing edges of $v$.
	Mechanism $f$ is \textit{impartial} on $\calG'\subseteq \calG$ if on this set of graphs the outgoing edges of a vertex have no influence on its selection, \ie if for every graph $G = (V,E)\in \calG'$, $v\in V$, and $G'\in \calN_v(G)$, it holds that $f(G)\cap\{v\}=f(G')\cap\{v\}$.
	Given a $k$-selection mechanism $f$ and an aggregator function $\sigma: 2^{\RR}\to \RR$ such that $\sigma(\emptyset)=0$ and, for every $S\subseteq \RR$ with $|S|\geq 1$, $\min\{x\in S\}\leq \sigma(S)\leq \max\{x\in S\}$, we say that $f$ is \textit{$\alpha$-$\sigma$-additive} on $\calG'\subseteq\calG$, for $\alpha\geq 0$, if for every graph in $\calG'$ the function $\sigma$ evaluated on the choice of $f$ differs from the maximum indegree by at most $\alpha$, \ie if
	\[
	\sup_{\substack{G\in \calG'}} \Bigl\{ \Delta(G)-\sigma\left(\{\delta^-(v, G)\}_{v\in f(G)}\right) \Bigr\} \leq \alpha.
	\]
	We will specifically be interested in the cases where $\sigma$ is the minimum, the median, and the mean, and respectively call a mechanism \textit{$\alpha$-min-additive}, \textit{$\alpha$-median-additive}, and \textit{$\alpha$-mean-additive} in these cases.
	
	\section{Plurality with Runners-up}
	\label{sec:pwru}
	
	Focusing on the case with maximum outdegree one, \citet{TaOh14a} proposed a mechanism they called \textit{plurality with runners-up}. The mechanism, which we describe formally in \autoref{alg:PWRU}, selects all vertices with maximum indegree; if there is a unique such vertex, then any vertex with an outgoing edge to that vertex whose indegree is smaller by one is selected as well. The idea behind this mechanism is that vertices in the latter category would be among those with maximum degree if their outgoing edge was deleted, and thus any impartial mechanism seeking to select the vertices with maximum degree would also have to select those vertices. Plurality with runners-up is impartial on $\calG(1)$, and in any graph with $n$ vertices selects between $1$ and $n$ vertices whose degree is equal to the maximum degree or the maximum degree minus one. It is thus an impartial and $1$-min-additive $n$-selection mechanism on $\calG_n(1)$ for every $n\in \NN$. It is natural to ask whether a similar additive guarantee can be obtained for more general settings.
	In this section, we answer this question in the affirmative, and in particular study for which values of $n$, $k$, and $d$ there exists an impartial and $1$-min-additive $k$-selection mechanism on $\calG_n(d)$. We will see later, in \autoref{sec:imp-results}, that $1$-min-additivity is in fact best possible for all cases covered by our result, with the exception of the boundary case where $n=2$.
	\begin{algorithm}[t]
		\SetAlgoNoLine
		\KwIn{Digraph $G=(V,E)\in \calG_n(1)$}
		\KwOut{Set $S\subseteq V$ of selected vertices}
		Let $S=\{v\in V: \delta^-(v)=\Delta(G)\}$\;
		\If{$S=\{v\}$ for some $v\in V$}{
			$S\xleftarrow{} S\cup \{u\in V: \delta^-(u)=\Delta(G)-1 \text{ and } (u,v)\in E\}$
		}
		{\bf Return} $S$
		\caption{Plurality with runners-up}
		\label{alg:PWRU}
	\end{algorithm}
	
	While \citeauthor{TaOh14a} do not limit the maximum number of selected vertices, they discuss briefly a modification of their mechanism that retains impartiality and $1$-min-additivity but selects at most $2$ vertices. Instead of all vertices with maximum indegree, the modified mechanism breaks ties in favor of a single maximum-degree vertex using a fixed ordering of the vertices. In order to guarantee impartiality, the modified mechanism then also selects any vertex that would be selected in the graph obtained by deleting the outgoing edge of that vertex. The assumption that every vertex has at most one outgoing edge means that at most one additional vertex is selected. There thus exists a $1$-min-additive $k$-selection mechanism on $G(1)$ for every $k\geq 2$.
	
	Our first result generalizes this mechanism to settings with arbitrary outdegrees, as long as the maximum number of selected vertices is large enough. To this end we show that when the maximum outdegree is $d$, to achieve impartiality, at most $d$ vertices have to be selected in addition to the one with maximum indegree and highest priority.\footnote{In this mechanism and wherever ties are broken in the rest of the paper, we break ties in favor of greater index, so $\text{top}(G)$ is the vertex with maximum indegree and highest priority in graph $G$. Naturally, any other deterministic tie-breaking rule could be used instead.} 
	\begin{algorithm}[t]
		\SetAlgoNoLine
		\KwIn{Digraph $G=(V,E)\in \calG_n$}
		\KwOut{Set $S\subseteq V$ of selected vertices}
		Let $S=\emptyset$\;
		\For{$v\in V$}{
			Let $G_v = (V,E\setminus (\{v\}\times V))$\;
			\If{$\ttop(G_v)=v$}{
				$S\xleftarrow{} S\cup \{v\}$
			}
		}
		{\bf Return} $S$
		\caption{Asymmetric plurality with runners-up $\calP(G)$}
		\label{alg:PWRU-order}
	\end{algorithm}
	We formally describe the resulting mechanism in \autoref{alg:PWRU-order}, and will refer to it as \textit{asymmetric plurality with runners-up} and denote its output on graph $G$ by $\calP(G)$. We obtain the following theorem, which generalizes the known result for the outdegree-one case.
	\begin{theorem}
		\label{thm:ub-1}
		For every $n\in \NN,\ d\in[n-1]$, and $k\in\{d+1,\ldots,n\}$, there exists an impartial and $1$-min-additive $k$-selection mechanism on $\calG_n(d)$.
	\end{theorem}
	
	We will be interested in the following in comparing vertices both according to their indegree and to their index, and we will use regular inequality symbols, as well as the operators $\max$ and $\min$, to denote the lexicographic order among pairs of the form $(\delta^-(v),v)$.
	The following lemma characterizes the structure of the set of vertices selected by \autoref{alg:PWRU-order}, and provides the main technical ingredient to the proof of \autoref{thm:ub-1}.
	\begin{lemma}
		\label{lem:char-pwru}
		Let $G=(V,E)\in \calG$
		and $v\in V$. Then, $v\in \calP(G)$ if and only if
		\begin{enumerate}[label=(\alph*)]	
			\item for every $w\in V$ with $(\delta^-(w),w)>(\delta^-(v),v)$ it holds $(v,w)\in E$; and \label{item:outgoing-edges}
			\item one of the following holds:\label{item:indegree}
			\begin{enumerate}[label=(\roman*)]	
				\item $\delta^-(v)=\Delta(G)$; or\label{alt-max-indegree}
				\item $\delta^-(v)=\Delta(G)-1$ and $v>w$ for every $w\in V$ with $\delta^-(w)=\Delta(G)$.\label{alt-max-indegree-minus-1}
			\end{enumerate}
		\end{enumerate}
	\end{lemma}
	\begin{proof}
		We first show that, if $v\in \calP(G)$ for a given graph $G$, then~\ref{item:outgoing-edges} and~\ref{item:indegree} follow.
		Let $G=(V,E)\in \calG$, and let $v\in \calP(G)$. 
		To see~\ref{item:outgoing-edges}, suppose there is $w\in V$ with $(\delta^-(w,G),w)>(\delta^-(v,G),v)$.
		Since $v\in \calP(G)$, we have $v=\ttop(G_v)$ with $G_v = (V,E\setminus (\{v\}\times V))$. 
		This implies $(\delta^-(v,G_v),v)>(\delta^-(w,G_v),w)$ and therefore $\delta^-(w,G)>\delta^-(w,G_v)$, because $\delta^-(v,G)=\delta^-(v,G_v)$.
		Since $G$ and $G_v$ only differ in the outgoing edges of $v$, we conclude that $(v,w)\in E$.
		To prove~\ref{item:indegree}, we note that for every $w\in V$ we have
		\begin{equation}
			(\delta^-(v,G),v)=(\delta^-(v,G_v),v)>(\delta^-(w,G_v),w)\geq (\delta^-(w,G)-1,w),\label{eq:indegrees-pwru}
		\end{equation}
		where the last inequality comes from the fact that each vertex has at most one incoming edge from $v$.
		If there is no $w\in V\setminus \{v\}$ with $\delta^-(w)=\Delta(G)$, the maximum indegree must be that of $v$, so $\delta^-(v)=\Delta(G)$ and~\ref{alt-max-indegree} follows. Otherwise, for each $w\in V\setminus \{v\}$ with $\delta^-(w)=\Delta(G)$,~\autoref{eq:indegrees-pwru} yields $(\delta^-(v,G),v)>(\Delta(G)-1,w)$.
		We conclude that either $\delta^-(v,G)>\Delta(G)-1$, in which case~\ref{alt-max-indegree} holds, or both $\delta^-(v)=\Delta(G)-1$ and $v>w$, which implies~\ref{alt-max-indegree-minus-1}.
		
		We now prove the other direction. Let $G=(V,E)\in \calG$ and $v\in V$ such that both~\ref{item:outgoing-edges} and~\ref{item:indegree} hold.
		Let $G_v=(V,E\setminus (\{v\}\times V))$.
		We have to show that $\ttop(G_v)=v$, \ie that for every $w\in V\setminus \{v\},\ (\delta^-(v, G_v), v)>(\delta^-(w,G_v),w)$.
		Let $w$ be a vertex in $V\setminus \{v\}$.
		If $(\delta^-(v, G), v)>(\delta^-(w,G),w)$, we can conclude immediately since $\delta^-(v, G_v)=\delta^-(v, G)$ and $\delta^-(w,G_v)\leq \delta^-(w,G)$.
		Otherwise, we know from~\ref{item:outgoing-edges} that $(v,w)\in E$ and thus $\delta^-(w,G_v)= \delta^-(w,G)-1$.
		If $v$ satisfies~\ref{alt-max-indegree}, this yields
		\[
		\delta^-(v, G_v) = \delta^-(v, G) = \Delta(G) \geq \delta^-(w,G) = \delta^-(w,G_v)+1,
		\]
		so $(\delta^-(v, G_v), v)>(\delta^-(w,G_v),w)$.
		On the other hand, if $v$ satisfies~\ref{alt-max-indegree-minus-1}, then 
		\[
		\delta^-(v, G_v) = \delta^-(v, G) = \Delta(G)-1 \geq \delta^-(w,G)-1 = \delta^-(w,G_v),
		\]
		and $v>w$ implies $(\delta^-(v, G_v), v)>(\delta^-(w,G_v),w)$ as well.
	\end{proof}
	
	Observe that \autoref{lem:char-pwru} implies in particular that $\ttop(G)\in \calP(G)$ for every graph $G$.
	\autoref{fig:example-pwru} provides an example of the characterization given by \autoref{lem:char-pwru}, in terms of indegrees, tie-breaking order, and edges among selected vertices.
	\begin{figure}[t]
		\centering
		\begin{tikzpicture}
			\Vertex[x=3, y=2, Math, shape=circle, color=white, size=.3, label=4, fontscale=1.2, position=above, distance=-.04cm]{A}
			\Vertex[x=6, y=2, Math, shape=circle, color=black, size=.3, label=2, fontscale=1.2, position=above, distance=-.04cm]{B}
			\Vertex[x=7.5, y=2, Math, shape=circle, color=white, size=.3, label=1, fontscale=1.2, position=above, distance=-.04cm]{C}
			\Vertex[Math, shape=circle, color=black, size=.3, label=6, fontscale=1.2, position=below, distance=-.04cm]{D}
			\Vertex[x=1.5, Math, shape=circle, color=white, size=.3, label=5, fontscale=1.2, position=below, distance=-.04cm]{E}
			\Vertex[x=4.5, Math, shape=circle, color=black, size=.3, label=3, fontscale=1.2, position=below, distance=-.04cm]{F}
			\Edge[Direct, color=black, lw=1pt, bend=-28](C)(A)
			\Edge[Direct, color=black, lw=1pt](C)(B)
			\Edge[Direct, color=black, lw=1pt, bend=10](D)(A)
			\Edge[Direct, color=black, lw=1pt, bend=10](D)(B)
			\Edge[Direct, color=black, lw=1pt](E)(A)
			\Edge[Direct, color=black, lw=1pt, bend=-10](E)(B)
			\Edge[Direct, color=black, lw=1pt, bend=-10](E)(C)
			\Edge[Direct, color=black, lw=1pt](E)(D)
			\Edge[Direct, color=black, lw=1pt](F)(A)
			\Edge[Direct, color=black, lw=1pt, bend=-10](F)(B)
			\Edge[Direct, color=black, lw=1pt, bend=-20](F)(C)
			\Edge[Direct, color=black, lw=1pt, bend=28](F)(D)
			\Edge[Direct, color=black, lw=1pt, bend=10](F)(E)
			\Text[x=-1, y=2, fontsize=\small]{$\Delta$};
			\Text[x=-1, fontsize=\small]{$\Delta-1$};
		\end{tikzpicture}
		\caption{Example of a set of vertices selected by \autoref{alg:PWRU-order}. In this illustration and throughout the paper, vertices are arranged vertically according to indegree and horizontally according to index, so that vertices on the left are favored in case of ties. The vertices selected by the mechanism are drawn in white, those not selected in black. Vertices with indegree below $\Delta-1$, as well as edges incident to such vertices, are not shown. Denoting the graph as $G=(V,E)$, and letting $G_v = (V,E\setminus (\{v\}\times V))$ for each vertex $v$, the selected vertices $v$ are those for which $\ttop(G_v)=v$. Specifically, vertices $2$, $3$, and $6$ are not selected because $\ttop(G_2)=4$, $\ttop(G_3)=4$, and $\ttop(G_6)=1$.}
		\label{fig:example-pwru}
	\end{figure}
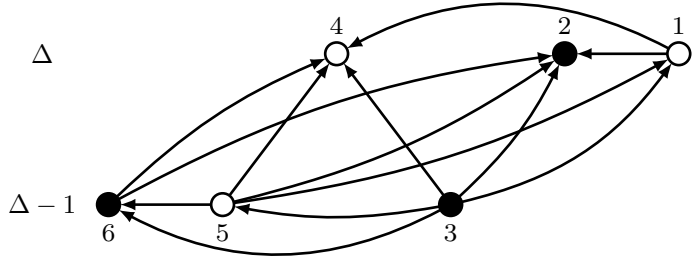
	
	We are now ready to prove \autoref{thm:ub-1}.
	\begin{proof}[Proof of \autoref{thm:ub-1}]
		We show that for every $n\in \NN$ and $d\in [n-1]$, asymmetric plurality with runners-up is impartial and $1$-min-additive on $\calG_n(d)$, and that for every $G=(V,E)\in \calG_n(d)$, it selects at most $d+1$ vertices. If this is the case, then for every $k\in \{d+1,\ldots,n\}$ the mechanism would satisfy the statement of the theorem. Therefore, let $n$ and $d$ be as mentioned.
		
		Impartiality follows from the definition of the mechanism, because the outgoing edges of a vertex are not taken into account when deciding whether the vertex is taking part on the selected set or not. If we let $G=(V,E),\ v\in V$, and $G'=(V,E')\in \calN_v(G)$, then the graphs $G_v$ and $G'_v$ constructed when running the mechanism with each of these graphs $G$ and $G'$ as an input, respectively, are the same because by definition of $\calN_v(G)$ we have $E\setminus(\{v\}\times V)=E'\setminus(\{v\}\times V)$. Since $v\in \calP(G)\Leftrightarrow \ttop(G_v)=v$, and $v\in \calP(G')\Leftrightarrow \ttop(G'_v)=v$, we conclude $v\in \calP(G) \Leftrightarrow v\in \calP(G')$.
		
		To see that the mechanism is $1$-min-additive, let $G\in \calG_n(d)$ and first note that $\calP(G)\neq\emptyset$ since \autoref{lem:char-pwru} implies that $\ttop(G)\in \calP(G)$. From this lemma we also know that for every $v\in \calP(G),\ \delta^-(v)\geq \Delta(G)-1$. We conclude that $\min \{\{\delta^-(v)\}_{v\in \calP(G)}\} \geq \Delta(G)-1$, and since this holds for every $G\in \calG_n(d)$, the mechanism is $1$-min-additive.
		
		Finally, let $G=(V,E)\in \calG_n(d)$, and suppose that $|\calP(G)|>d+1$. 
		If we denote $v_L=\text{argmin}_{v\in \calP(G)}\{(\delta^-(v),v)\}$, from \autoref{lem:char-pwru} we know that $(v_L,w)\in E$ for every $w\in V$ with $(\delta^-(w),w)>(\delta^-(v_L),v_L)$, thus $\delta^+(v_L)\geq |\calP(G)|-1 > d$, a contradiction. We conclude that $|\calP(G)|\leq d+1$.
	\end{proof}
	
	The following result, concerning mechanisms that may select an arbitrary number of vertices, follows immediately from \autoref{thm:ub-1}.
	\begin{corollary}
		For every $n\in\NN$, there exists an impartial and $1$-min-additive $n$-selection mechanism on $\calG_n$.
	\end{corollary}

	On $\calG_n$, \ie in the case of unbounded outdegrees, this result can in fact be improved slightly to guarantee $1$-min-additivity while selecting only at most $n-1$ vertices. The improvement is achieved by a more intricate version of asymmetric plurality with runners-up, which we call \textit{asymmetric plurality with runners-up and pivotal vertices}. We formally describe this mechanism in \autoref{alg:PWRU-pivotal} and denote its output for graph $G$ by $\calP^P(G)$. Given a graph $G=(V,E)$, call a vertex $u\in \calP(G)$ \textit{pivotal} for $v\in\calP(G)$ if there exists a graph $G_{uv}\in \calN_u(G)$ such that $v\notin \calP(G_{uv})$, \ie if the outgoing edges of $u$ can be changed in such a way that $v$ is no longer selected by asymmetric plurality with runners-up. Asymmetric plurality with runners-up and pivotal vertices then selects every vertex in $\calP(G)$ that is pivotal for every other vertex in $\calP(G)$. The mechanism turns out to inherit impartiality and $1$-min-additivity, and to never select all vertices.
	\begin{algorithm}[t]
		\SetAlgoNoLine
		\KwIn{Digraph $G=(V,E)\in \calG_n$}
		\KwOut{Set $S\subseteq V$ of selected vertices with $|S|\leq n-1$}
		Let $S\xleftarrow{} \emptyset$\;
		\For{$u\in \calP(G)$}{
			\If{for every $v\in \calP(G)\setminus \{u\}$ there exists $G_{uv}\in \calN_u(G)$ such that $v\notin \calP(G_{uv})$}{
				$S\xleftarrow{} S\cup \{u\}$
			}
		}
		{\bf Return} $S$
		\caption{Asymmetric plurality with runners-up and pivotal vertices $\mathsf{\calP^P(G)}$}
		\label{alg:PWRU-pivotal}
	\end{algorithm}
	
	\begin{theorem}\label{thm:ub-1-n-1}
		For every $n\in \NN$ and $k\in \{n-1,n\}$, there exists an impartial and $1$-min-additive $k$-selection mechanism on $\calG_n$.
	\end{theorem}
	
	\begin{proof}
		We show that for every $n\in \NN$, asymmetric plurality with runners-up and pivotal vertices is impartial and $1$-min-additive on $\calG_n$ and that for every $G=(V,E)\in \calG_n$, it selects at most $n-1$ vertices. Let $n\in \NN$ be an arbitrary value.
		
		To see that the mechanism is impartial, let $G=(V,E)\in \calG_n,\ u\in \calP^P(G)$, and $G'=(V,E')\in \calN_u(G)$.
		We show in the following that $u\in \calP^P(G')$, and since the graphs $G$ and $G'$ are chosen arbitrarily, their roles can be inverted and this is enough to conclude that the mechanism is impartial.
		We first note that $u\in \calP(G)$ because $\calP^P(G)\subseteq \calP(G)$, thus impartiality of asymmetric plurality with runners-up proven in \autoref{thm:ub-1} implies $u\in \calP(G')$.
		If $\calP(G')=\{u\}$, then the condition in the mechanism holds trivially for this vertex, so $u\in \calP^P(G')$ and we conclude.
		Otherwise, let $v\in \calP(G')\setminus\{u\}$ be an arbitrary vertex selected by asymmetric plurality with runners-up other than $u$.
		Since $u\in \calP^P(G)$, we have that either (a)~$v\notin \calP(G)$, or (b)~$v\in \calP(G)$ and there exists $G_{uv} = (V, E_{uv}) \in \calN_u(G)$ such that $v\notin \calP(G_{uv})$.
		If (a)~holds, taking $G'_{uv}=G$, which belongs to $\calN_u(G')$ because of the assumption that $G' \in \calN_u(G)$, we have that $v\notin \calP(G'_{uv})$.
		If (b)~holds, taking $G'_{uv}=G_{uv}$, which belongs to $\calN_u(G')$ since $\calN_u(G')=\calN_u(G)$, we have that $v\notin \calP(G'_{uv})$. 
		In either case, we conclude that there exists $G'_{uv}\in \calN_u(G')$ such that $v\notin \calP(G'_{uv})$.
		Since this argument is valid for every $v\in \calP(G')\setminus\{u\}$, we conclude that $u\in \calP^P(G')$.
		
		To see that the mechanism is $1$-min-additive, it is enough to show that it always selects a vertex, since for every $G\in \calG$ it selects a subset of $\calP(G)$ and from \autoref{thm:ub-1} we know that this set contains vertices with indegrees in $\{\Delta(G), \Delta(G)-1\}$.
		To this purpose we let $G=(V,E)\in \calG_n$ and introduce some additional notation. 
		Let $S^i=\{v\in \calP(G):\delta^-(v)=\Delta(G)-i\}$ and $n^i=|S^i|$ for $i\in \{0,1\}$, and denote
		\begin{align*}
			v_H&=\text{argmax}_{v\in \calP(G)}\{(\delta^-(v,G), v)\}=\ttop(G),\\
			v_L&=\text{argmin}_{v\in \calP(G)}\{(\delta^-(v,G), v)\}.
		\end{align*}
		From \autoref{lem:char-pwru}, we know that $\calP(G)= S^0\cup S^1$, that $(v_L,v)\in E$ for every $v\in \calP(G)\setminus\{v_L\}$, and that $u>v$ for each $u\in S^1,\ v\in S^0$.
		We now distinguish two cases according to the edges between vertices in $\calP(G)$.
		
		If $(v_H,v)\in E$ for every $v\in\calP(G)\setminus \{v_H\}$, then we claim that defining $G'=(V,E\setminus (\{v_H\}\times V))\in \calN_{v_H}(G)$ it holds $v\notin \calP(G')$ for every $v\in \calP(G)\setminus \{v_H\}$. 
		If this is true, it is clear that $v_H\in \calP^P(G)$ and thus $\calP^P(G)\neq\emptyset$. We now prove the claim. 
		First, note that $v_H\in \calP(G')$ since $v_H=\ttop(G')$ and \autoref{lem:char-pwru} ensures $\ttop(G')\in \calP(G')$. 
		This comes from the fact that $\delta^-(v_H,G')=\delta^-(v_H,G)$ and $\delta^-(v,G')\leq \delta^-(v,G)$ for every $v\in V\setminus \{v_H\}$, together with $v_H=\ttop(G)$. 
		Moreover, for every $v\in S^0\setminus \{v_H\}$ it holds $\delta^-(v,G')=\delta^-(v,G)-1 = \delta^-(v_H,G')-1= \Delta(G')-1$ and $v<v_H$, so condition~\ref{item:indegree} in \autoref{lem:char-pwru} does not hold for $v$ and thus $v\notin \calP(G')$.
		Analogously, for every $v\in S^1$ it holds $\delta^-(v,G')=\delta^-(v,G)-1 = \delta^-(v_H,G')-2= \Delta(G')-2$, so condition~\ref{item:indegree} in \autoref{lem:char-pwru} does not hold for $v$ either, and thus $v\notin \calP(G')$.
		This allows to conclude the claim and the fact that $\calP^P(G)$ is non-empty for this case.
		This argument is illustrated in \autoref{fig:thm-n-1-example1}.
		
		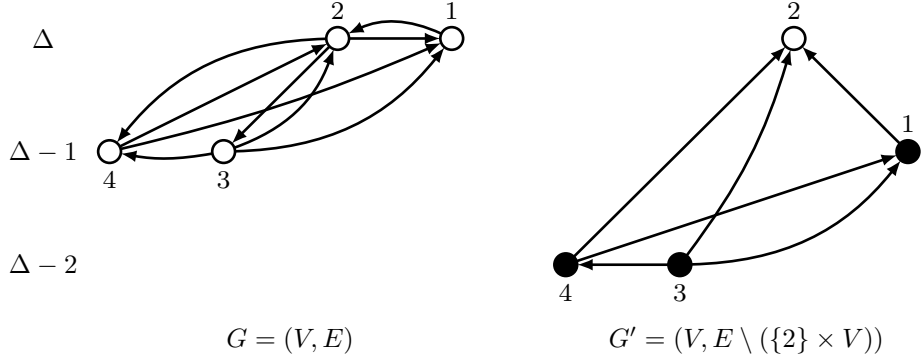
\begin{figure}[t]
			\centering
			\begin{tikzpicture}
				
				\Text[x=-1, y=3, fontsize=\small]{$\Delta$};
				\Text[x=-1, y=1.5, fontsize=\small]{$\Delta-1$};
				\Text[x=-1, fontsize=\small]{$\Delta-2$};
				\Text[x=2.25, y=-1, fontsize=\small]{$G=(V,E)$};
				\Text[x=8.25, y=-1, fontsize=\small]{$G'=(V, E\setminus (\{2\}\times V))$};
				
				\Vertex[x=3, y=3, Math, shape=circle, color=white, size=.3, label=2, fontscale=1.2, position=above, distance=-.04cm]{A}
				\Vertex[x=4.5, y=3, Math, shape=circle, color=white, size=.3, label=1, fontscale=1.2, position=above, distance=-.04cm]{B}
				\Vertex[y=1.5, Math, shape=circle, color=white, size=.3, label=4, fontscale=1.2, position=below, distance=-.04cm]{C}
				\Vertex[x=1.5, y=1.5, Math, shape=circle, color=white, size=.3, label=3, fontscale=1.2, position=below, distance=-.04cm]{D}
				\Edge[Direct, color=black, lw=1pt](A)(B)
				\Edge[Direct, color=black, lw=1pt, bend=-25](A)(C)
				\Edge[Direct, color=black, lw=1pt](A)(D)
				\Edge[Direct, color=black, lw=1pt, bend=-25](B)(A)
				\Edge[Direct, color=black, lw=1pt](C)(A)
				\Edge[Direct, color=black, lw=1pt, bend=-5](C)(B)
				\Edge[Direct, color=black, lw=1pt, bend=-25](D)(A)
				\Edge[Direct, color=black, lw=1pt, bend=-25](D)(B)
				\Edge[Direct, color=black, lw=1pt, bend=10](D)(C)
				
				\Vertex[x=9, y=3, Math, shape=circle, color=white, size=.3, label=2, fontscale=1.2, position=above, distance=-.04cm]{A}
				\Vertex[x=10.5, y=1.5, Math, shape=circle, color=black, size=.3, label=1, fontscale=1.2, position=above, distance=-.04cm]{B}
				\Vertex[x=6, Math, shape=circle, color=black, size=.3, label=4, fontscale=1.2, position=below, distance=-.04cm]{C}
				\Vertex[x=7.5, Math, shape=circle, color=black, size=.3, label=3, fontscale=1.2, position=below, distance=-.04cm]{D}
				\Edge[Direct, color=black, lw=1pt](B)(A)
				\Edge[Direct, color=black, lw=1pt](C)(A)
				\Edge[Direct, color=black, lw=1pt](C)(B)
				\Edge[Direct, color=black, lw=1pt, bend=-10](D)(A)
				\Edge[Direct, color=black, lw=1pt, bend=-25](D)(B)
				\Edge[Direct, color=black, lw=1pt](D)(C)
				
			\end{tikzpicture}
			\caption{Illustration of the fact that the set of vertices selected by \autoref{alg:PWRU-pivotal} is non-empty if $(v_H,v)\in E$ for every $v\in \calP(G)\setminus \{v_H\}$. Vertices selected by asymmetric plurality with runners-up are drawn in white. Denoting the graph on the left as $G=(V,E)$, where $v_H=2$, and defining $G'=(V,E\setminus (\{2\}\times V))\in \calN_2(G)$, we have that $\{1,3,4\}\cap \calP(G')=\emptyset$, and thus $2\in \calP^P(G)$.}
			\label{fig:thm-n-1-example1}
		\end{figure}
		
		Now we consider the case where there is a vertex $\bar{v}\in \calP(G)$ such that $(v_H, \bar{v})\notin E$, and we claim that defining $G'=(V,(E\setminus (\{v_L\times V\}))\cup (v_L,v_H))\in \calN_{v_L}(G)$ it holds $v\notin \calP(G')$ for every $v\in \calP(G)\setminus \{v_L, v_H\}$, whereas defining $G''=(V,E\setminus (v_L,v_H))\in \calN_{v_L}(G)$ it holds $v_H\notin \calP(G'')$. 
		If this is true, then $v_L\in \calP^P(G)$ and $\calP^P(G)\neq\emptyset$. 
		We now prove the claim.
		First, note that $v_H\in \calP(G')$ for the same reason as before, since $\delta^-(v_H,G')=\delta^-(v_H,G)$ and $\delta^-(v,G')\leq \delta^-(v,G)$ for every $v\in V\setminus \{v_H\}$. 
		Moreover, for every $v\in S^0\setminus \{v_H\}$ it holds $\delta^-(v,G')=\delta^-(v,G)-1 = \delta^-(v_H,G')-1= \Delta(G')-1$ and $v<v_H$, so condition~\ref{item:indegree} in \autoref{lem:char-pwru} does not hold for $v$ and thus $v\notin \calP(G')$.
		Analogously, for every $v\in S^1\setminus \{v_L\}$ it holds $\delta^-(v,G')=\delta^-(v,G)-1 = \delta^-(v_H,G')-2= \Delta(G')-2$ so condition~\ref{item:indegree} in \autoref{lem:char-pwru} does not hold for $v$ and thus $v\notin \calP(G')$.
		This allows to conclude the claim for $G'$. 
		In the case of $G''$, we can write the following chain of inequalities,
		\[
		(\delta^-(\bar{v},G''),\bar{v})=(\delta^-(\bar{v},G), \bar{v}) > (\delta^-(v_H,G)-1, v_H) = (\delta^-(v_H,G''), v_H),
		\]
		where the equalities hold because of the definition of $G''$ and the inequality by condition~\ref{item:indegree} in \autoref{lem:char-pwru}, given that $\bar{v}\in \calP(G)$.
		Since $(v_H,\bar{v})\notin E$, we conclude from condition~\ref{item:outgoing-edges} in \autoref{lem:char-pwru} that $v_H\notin \calP(G'')$, and therefore the claim for $G''$ follows.
		This argument is illustrated in \autoref{fig:thm-n-1-example2}.
		
		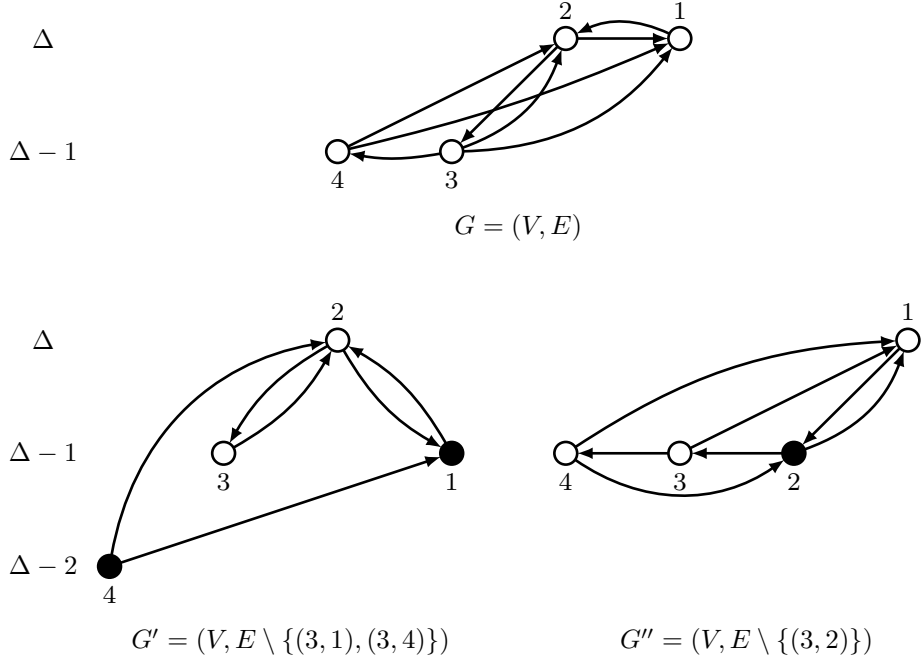
\begin{figure}[t]
			\centering
			\begin{tikzpicture}
				
				\Text[x=-1, y=3, fontsize=\small]{$\Delta$};
				\Text[x=-1, y=1.5, fontsize=\small]{$\Delta-1$};
				\Text[x=-1, fontsize=\small]{$\Delta-2$};
				\Text[x=2.25, y=-1, fontsize=\small]{$G'=(V,E\setminus \{(3,1), (3,4)\})$};
				\Text[x=8.25, y=-1, fontsize=\small]{$G''=(V, E\setminus \{(3,2)\})$};
				
				\Text[x=-1, y=7, fontsize=\small]{$\Delta$};
				\Text[x=-1, y=5.5, fontsize=\small]{$\Delta-1$};
				\Text[x=5.25, y=4.5, fontsize=\small]{$G=(V,E)$};
				
				\Vertex[x=6, y=7, Math, shape=circle, color=white, size=.3, label=2, fontscale=1.2, position=above, distance=-.04cm]{A}
				\Vertex[x=7.5, y=7, Math, shape=circle, color=white, size=.3, label=1, fontscale=1.2, position=above, distance=-.04cm]{B}
				\Vertex[x=3, y=5.5, Math, shape=circle, color=white, size=.3, label=4, fontscale=1.2, position=below, distance=-.04cm]{C}
				\Vertex[x=4.5, y=5.5, Math, shape=circle, color=white, size=.3, label=3, fontscale=1.2, position=below, distance=-.04cm]{D}
				\Edge[Direct, color=black, lw=1pt](A)(B)
				\Edge[Direct, color=black, lw=1pt](A)(D)
				\Edge[Direct, color=black, lw=1pt, bend=-25](B)(A)
				\Edge[Direct, color=black, lw=1pt](C)(A)
				\Edge[Direct, color=black, lw=1pt, bend=-5](C)(B)
				\Edge[Direct, color=black, lw=1pt, bend=-25](D)(A)
				\Edge[Direct, color=black, lw=1pt, bend=-25](D)(B)
				\Edge[Direct, color=black, lw=1pt, bend=10](D)(C)

				\Vertex[x=3, y=3, Math, shape=circle, color=white, size=.3, label=2, fontscale=1.2, position=above, distance=-.04cm]{A}
				\Vertex[x=4.5, y=1.5, Math, shape=circle, color=black, size=.3, label=1, fontscale=1.2, position=below, distance=-.04cm]{B}
				\Vertex[Math, shape=circle, color=black, size=.3, label=4, fontscale=1.2, position=below, distance=-.04cm]{C}
				\Vertex[x=1.5, y=1.5, Math, shape=circle, color=white, size=.3, label=3, fontscale=1.2, position=below, distance=-.04cm]{D}
				\Edge[Direct, color=black, lw=1pt, bend=-15](A)(B)
				\Edge[Direct, color=black, lw=1pt, bend=-15](A)(D)
				\Edge[Direct, color=black, lw=1pt, bend=-15](B)(A)
				\Edge[Direct, color=black, lw=1pt, bend=35](C)(A)
				\Edge[Direct, color=black, lw=1pt](C)(B)
				\Edge[Direct, color=black, lw=1pt, bend=-15](D)(A)
				
				\Vertex[x=9, y=1.5, Math, shape=circle, color=black, size=.3, label=2, fontscale=1.2, position=below, distance=-.04cm]{A}
				\Vertex[x=10.5, y=3, Math, shape=circle, color=white, size=.3, label=1, fontscale=1.2, position=above, distance=-.04cm]{B}
				\Vertex[x=6, y=1.5, Math, shape=circle, color=white, size=.3, label=4, fontscale=1.2, position=below, distance=-.04cm]{C}
				\Vertex[x=7.5, y=1.5, Math, shape=circle, color=white, size=.3, label=3, fontscale=1.2, position=below, distance=-.04cm]{D}
				\Edge[Direct, color=black, lw=1pt, bend=-25](A)(B)
				\Edge[Direct, color=black, lw=1pt](A)(D)
				\Edge[Direct, color=black, lw=1pt](B)(A)
				\Edge[Direct, color=black, lw=1pt, bend=-35](C)(A)
				\Edge[Direct, color=black, lw=1pt, bend=15](C)(B)
				\Edge[Direct, color=black, lw=1pt](D)(B)
				\Edge[Direct, color=black, lw=1pt](D)(C)
				
			\end{tikzpicture}
			\caption{Illustration of the fact that the set of vertices selected by \autoref{alg:PWRU-pivotal} is non-empty if $(v_H,\bar{v})\notin E$ for some $\bar{v}\in \calP(G)\setminus \{v_H\}$. Vertices selected by asymmetric plurality with runners-up are drawn in white. Denoting the graph at the top by $G=(V,E)$, where $v_H=2$, $v_L=3$, and $\bar{v}=4$, and defining $G'=(V,E\setminus \{(3,1), (3,4)\})\in \calN_3(G)$, we have that $\{1,4\}\cap \calP(G')=\emptyset$, whereas defining $G''=(V,(E\setminus \{(3,2)\})\in \calN_3(G)$ we have that $2\notin \calP(G'')$. We conclude that $3\in \calP^P(G)$.}
			\label{fig:thm-n-1-example2}
		\end{figure}
		
		Finally, we show that the mechanism selects at most $n-1$ vertices. Let $G=(V,E)\in \calG_n$.
		Since $\calP^P(G)\subseteq \calP(G)$, if $|\calP(G)|\leq n-1$ this is immediate.
		We thus suppose in what follows that $|\calP(G)|=n$.
		In particular, \autoref{lem:char-pwru} implies $(v,v_H)\in E$ for every $v\in V\setminus \{v_H\}$, thus $\Delta(G)=n-1$, and $\delta^-(v)\geq n-2$ for every $v\in V$.
		If $S^1=\emptyset$, then $\delta^-(v)=n-1$ for every $v\in V$, \ie $G$ is the complete graph.
		In this case, $v_H=n$ and we claim that $v\notin \calP^P(G)$ for each $v\in V\setminus \{n\}$, thus $|\calP^P(G)|\leq 1$. This comes from the fact that, for every $v\in V\setminus\{n\}$ and every $G'=(V,E')\in \calN_v(G)$ it holds $n\in \calP(G')$.
		To see this, note that $(n,v)\in E'$ for every $v\in V\setminus \{n\},\ \delta^-(n,G')\geq n-2=\Delta(G')-1$, and $n>v$ for every $v\in V\setminus \{n\}$, so \autoref{lem:char-pwru} ensures $n\in \calP(G')$.
		If $S_1\neq\emptyset$, then there is at least one vertex with outdegree less or equal than $n-2$.
		Let $u$ be an arbitrary vertex with $\delta^+(u)\leq n-2$, and let $\bar{v}\in S_1$ be the vertex with highest index such that $(u,\bar{v})\notin E$, \ie $\bar{v}=\max \{V\setminus N^+(u)\}$.
		Since $u\in \calP^P(G)$, there exists $G'=(V,E')\in \calN_u(G)$ such that $\bar{v}\notin \calP(G')$.
		From \autoref{lem:char-pwru}, this implies that there exists $\bar{w}\in V$ such that either (a)~$(\delta^-(\bar{w},G'),w)>(\delta^-(\bar{v},G'),\bar{v})$ and $(\bar{v},\bar{w})\notin E'$, or (b)~$\delta^-(\bar{w},G')>\delta^-(\bar{v},G')$ and $\bar{w}>\bar{v}$. Since $\bar{v}\in \calP(G)$, we know from this same lemma that if (a)~holds, $(\delta^-(\bar{w},G),w)<(\delta^-(\bar{v},G),\bar{v})$ because of having $\bar{w}\notin N^+(\bar{v},G)=N^+(\bar{v},G')$; and similarly, if (b)~holds, $\delta^-(\bar{w},G)\leq \delta^-(\bar{v},G)$ because of having $\bar{w}>\bar{v}$. In either case, since $\delta^-(\bar{v},G)\leq \delta^-(\bar{v},G')$, we conclude that $\delta^-(\bar{w},G')> \delta^-(\bar{w},G)$, and therefore $(u,\bar{w})\notin E$.
		If (a)~holds, this is a contradiction because we would have $\{u,\bar{v}\}\cap N^-(\bar{w},G)=\emptyset$ and thus $\delta^-(\bar{w},G)\leq n-3$.
		If (b)~holds, we reach a contradiction as well, because we would have $\bar{w}\in V\setminus N^+(u,G)$ and $\bar{w}>\bar{v}$, but we chose $\bar{v}$ to be the maximum of this set.
	\end{proof}

	\section{An Impossibility Result}
	\label{sec:imp-results}
	
	When we established the existence of an impartial and $1$-min-additive $k$-selection mechanism on $\calG(d)$ whenever $k\geq d+1$, we claimed this result to be best possible in the sense that the additive guarantee cannot be improved. We will prove this claim, that impartiality is incompatible with the requirement to only select vertices with maximum indegree, as a corollary of a more general result.
	
	While selecting only vertices with maximum indegree is a natural goal for mechanisms that select varying numbers of vertices, other natural objectives exist for such mechanisms such as maximizing the median or mean indegree of the selected vertices. For both of these objectives, the mechanisms discussed in the previous section immediately provide upper bounds: if a $k$-selection mechanism always selects one vertex with maximum indegree and is $\alpha$-min-additive then it is clearly $\alpha$-median-additive and $\left(\frac{k-1}{k}\alpha\right)$-mean-additive; \autoref{thm:ub-1} thus implies the existence of a $1$-median-additive and $\frac{k-1}{k}$-mean-additive $k$-selection mechanism on $\calG(d)$, whenever $k\geq d+1$. To improve on $1$-median-additivity, it would be acceptable to select vertices with low indegree as long as a greater number of vertices with maximum indegree is selected at the same time. To improve on $\frac{k-1}{k}$-mean-additivity, it would suffice to select more than one vertex with maximum indegree whenever this is possible, and to otherwise select only a sublinear number in $k$ of vertices with indegree equal to the maximum indegree minus one. The following result shows that no such improvements are possible.
	\begin{theorem}  \label{thm:impossibility}
		Let $n\in\NN$, $n\geq 3$, $k\in[n]$, and $d\in[n-1]$. Let $f$ be an impartial $k$-selection mechanism.
		If $f$ is $\alpha_1$-median-additive on $\calG_n(d)$, then $\alpha_1\geq 1/2(1+\mathds{1}(d\geq 3))$. 
		If $f$ is $\alpha_2$-mean-additive on $\calG_n(d)$, then $\alpha_2\geq \left\lfloor \frac{d+1}{2} \right\rfloor / \left(\left\lfloor \frac{d+1}{2} \right\rfloor + 1 \right)$.
	\end{theorem}
	\begin{proof}
		Let $n$, $k$, and $d$ be as in the statement of the theorem. In the following we suppose that there is an impartial $k$-selection mechanism $f$ which is either $\alpha_1$-median-additive on $\calG_n(d)$ with $\alpha_1< 1/2(1+\mathds{1}(d\geq 3))$, or $\alpha_2$-mean-additive on $\calG_n(d)$ with $\alpha_2 < \left\lfloor \frac{d+1}{2} \right\rfloor / \left(\left\lfloor \frac{d+1}{2} \right\rfloor + 1 \right)$.
		
		We first prove the result for the case $d=1$. We consider the graph $G=(V,E)\in \calG_n(1)$ with $E=\{(1,2),(2,3),(3,1)\}$, consisting of a $3$-cycle and $n-3$ isolated vertices. We consider as well, for $v\in \{1,2,3\}$, the graph $G_v=(V,E_v)$ where $v$ deviates from the $3$-cycle by changing its outgoing edge to the previous vertex in the cycle, \ie
		\[
		E_1=\{(1,3),(2,3),(3,1)\}, E_2=\{(1,2),(2,1),(3,1)\}, E_3=\{(1,2),(2,3),(3,2)\}.
		\]
		Since $f$ is $\alpha_1$-median-additive with $\alpha_1<1/2$ or $\alpha_2$-mean-additive with $\alpha_2<1/2$, we have that $f(G_1)=\{3\}$, $f(G_2)=\{1\}$, and $f(G_3)=\{2\}$.
		In particular, for $v\in \{1,2,3\},\ v\notin f(G_v)$.
		Since for each $v\in \{1,2,3\}$ it holds $E_v\setminus (\{v\}\times V)=E\setminus (\{v\}\times V)$, we conclude by impartiality that $v\notin f(G)$, and thus $f(G)\cap \{1,2,3\}=\emptyset$. This implies that both the median and the mean indegree of the vertices in $f(G)$ are $0$, which contradicts the additive guarantee of this mechanism because $\Delta(G)=1$.
		
		In the following, we assume $d\geq 2$.
		We denote $D=[d+1]$ and consider in what follows two families of graphs with $n$ vertices, $K_v$ for each $v\in D$ and $K_{uv}$ for each $u,v\in D, u\neq v$. They are constructed from a complete subgraph on $D$ but deleting the outgoing edges of $v$, in the case of $K_v$, and the outgoing edges of $u$ and $v$, in the case of $K_{uv}$. All the other vertices remain isolated. Formally, taking $V=[n]$ we define
		\begin{align*}
			K_v &= (V, (D\setminus \{v\})\times D) \text{ for every } v\in D,\\
			K_{uv} & = (V, (D\setminus \{u,v\})\times D) \text{ for every } u,v\in D \text{ with } u\neq v.
		\end{align*}
		
		If there is $v\in D$ such that $v\notin f(K_v)$, then
		\begin{align*}
			\text{median}\left(\{\delta^-(w, K_v)\}_{w\in f(K_v)}\right) &\leq d-1 = \Delta(K_v)-1,\\
			\text{mean}\left(\{\delta^-(w, K_v)\}_{w\in f(K_v)}\right) &\leq d-1 = \Delta(K_v)-1,
		\end{align*}
		which is a contradiction, so the result follows immediately.
		Therefore, in the following we assume that for every $v\in D$ we have $v\in f(K_v)$.
		We claim that for every $v\in D$,
		\[
		|\{u\in D\setminus \{v\}: u\in f(K_v)\}| \geq \left\lfloor\frac{d+1}{2}\right\rfloor.
		\]
		Let us see why the result follows if the claim holds. If this is the case, $f$ selects one vertex with maximum indegree $d$ in $K_v$ and at least $\left\lfloor\frac{d+1}{2}\right\rfloor$ vertices with indegree $d-1$. This yields both
		\[
		\text{median}\left(\{\delta^-(w, K_v)\}_{w\in f(K_v)}\right) \leq
		\left\{ \begin{array}{ll}
			d-\frac{1}{2} & \text{if }d=2\\
			d-1 & \text{otherwise,}
			
		\end{array}
		\right.
		\]
		and 
		\[
		\text{mean}\left(\{\delta^-(w, K_v)\}_{w\in f(K_v)}\right) \leq \frac{d + (d-1)\left\lfloor\frac{d+1}{2}\right\rfloor}{\left\lfloor\frac{d+1}{2}\right\rfloor+1} = d - \frac{\left\lfloor\frac{d+1}{2}\right\rfloor}{\left\lfloor\frac{d+1}{2}\right\rfloor+1},
		\]
		which is a contradiction since $\Delta(K_v)=d$.
		
		Now we prove the claim. Suppose that for every $v\in D$ we have $v\in f(K_v)$ and
		\begin{align}
			|\{u\in D\setminus \{v\}: u\in f(K_v)\}| < \left\lfloor\frac{d+1}{2}\right\rfloor.\label{eq:select-few-non-top-voted}
		\end{align} 
		Let $v\in D$ and  $u\in D \setminus \{v\}$ such that $u\notin f(K_v)$. Observing that
		\[
		((D\setminus \{v\})\times D)\setminus(\{u\}\times V) = ((D\setminus \{u,v\})\times D)\setminus(\{u\}\times V),
		\]
		we obtain from impartiality that $u\notin f(K_{uv})$. From the bounds on $\alpha_1$ or $\alpha_2$ that $f$ satisfies by assumption, this mechanism has to select a vertex with maximum indegree in this graph; otherwise, both the median and the mean of the selected set would be at most $\Delta(K_{uv})-1$. Since $\delta^-(w)<\Delta(K_{uv})$ for every $w\notin \{u,v\}$, it holds $v \in f(K_{uv})$. Using impartiality once again, we conclude $v\in f(K_u)$. We have shown the following property:
		\begin{align}
			\text{For every } u,v\in D: u\notin f(K_v) \Longrightarrow v\in f(K_u).\label{eq:cross-relation-complete-subgraphs}
		\end{align}
		
		Consider now the graph $H=(D,F)$, where for each $u,v\in D$ with $u\neq v$,
		$(u,v)\in F$ if and only if $u\notin f(K_v)$.
		Property \eqref{eq:select-few-non-top-voted} implies that
		\[
		\delta^-(v,H) > d-\left\lfloor\frac{d+1}{2}\right\rfloor \Longleftrightarrow \delta^-(v,H) \geq d+1-\left\lfloor\frac{d+1}{2}\right\rfloor
		\]
		for each $v\in D$. In particular, there has to be a vertex $v^*\in D$ such that $\delta^+(v^*,H) \geq d+1-\lfloor(d+1)/(2)\rfloor$ as well. For this vertex we have
		\[
		\delta^+(v^*,H) + \delta^-(v^*,H) \geq 2\left(d+1-\left\lfloor\frac{d+1}{2}\right\rfloor\right) \geq d+1.
		\]
		Since $H$ has $d+1$ vertices, this implies the existence of $w^*\in D$ for which $\{(v^*,w^*),(w^*,v^*)\}\subset F$, \ie both $v^*\notin f(K_{w^*})$ and $w^*\notin f(K_{v^*})$. This contradicts \eqref{eq:cross-relation-complete-subgraphs}, so we conclude the proof of the claim and the proof of the theorem.
	\end{proof}
	
	\autoref{fig:lb_3} provides an illustration of \autoref{thm:impossibility} for the case where $n=3$, \autoref{fig:lb_4} for the case where $n=4$.
	\begin{figure}[t]
		\centering
		\begin{tikzpicture}
			
			\Vertex[x=5, y=5.46, Math, shape=circle, color=black, size=.3, label=1, fontscale=1.2, position=left, distance=-.01cm]{A}
			\Vertex[x=4, y=3.73, Math, shape=circle, color=black, size=.3, label=2, fontscale=1.2, position=below, distance=-.01cm]{B}
			\Vertex[x=6, y=3.73, Math, shape=circle, color=white, size=.3, label=3, fontscale=1.2, position=below, distance=-.01cm]{C}
			\Edge[Direct, color=black, lw=1pt, bend=-15](A)(C)
			\Edge[Direct, color=black, lw=1pt](B)(C)
			\Edge[Direct, color=black, lw=1pt, bend=-15](C)(A)
			
			\draw[->] (5,3.23) -- (5,2.23);
			\Text[x=4.8, y=2.73, fontsize=\small]{$1$};

			\Vertex[x=1, y=1.73, Math, shape=circle, color=white, size=.3, label=1, fontscale=1.2, position=left, distance=-.01cm]{A}
			\Vertex[Math, shape=circle, color=black, size=.3, label=2, fontscale=1.2, position=below, distance=-.01cm]{B}
			\Vertex[x=2, Math, shape=circle, color=black, size=.3, label=3, fontscale=1.2, position=below, distance=-.01cm]{C}
			\Edge[Direct, color=black, lw=1pt, bend=-15](A)(B)
			\Edge[Direct, color=black, lw=1pt](C)(A)
			\Edge[Direct, color=black, lw=1pt, bend=-15](B)(A)
			
			\draw[->] (2.5,.86) -- (3.5,.86);
			\Text[x=3, y=1.06, fontsize=\small]{$2$};

			\Vertex[x=9, y=1.73, Math, shape=circle, color=black, size=.3, label=1, fontscale=1.2, position=left, distance=-.01cm]{A}
			\Vertex[x=8, Math, shape=circle, color=white, size=.3, label=2, fontscale=1.2, position=below, distance=-.01cm]{B}
			\Vertex[x=10, Math, shape=circle, color=black, size=.3, label=3, fontscale=1.2, position=below, distance=-.01cm]{C}
			\Edge[Direct, color=black, lw=1pt](A)(B)
			\Edge[Direct, color=black, lw=1pt, bend=-15](C)(B)
			\Edge[Direct, color=black, lw=1pt, bend=-15](B)(C)
			
			\draw[->] (7.5,.86) -- (6.5,.86);
			\Text[x=7, y=1.06, fontsize=\small]{$3$};

			\Vertex[x=5, y=1.73, Math, shape=circle, color=black, size=.3, label=1, fontscale=1.2, position=left, distance=-.01cm]{A}
			\Vertex[x=4, Math, shape=circle, color=black, size=.3, label=2, fontscale=1.2, position=below, distance=-.01cm]{B}
			\Vertex[x=6, Math, shape=circle, color=black, size=.3, label=3, fontscale=1.2, position=below, distance=-.01cm]{C}
			\Edge[Direct, color=black, lw=1pt](A)(B)
			\Edge[Direct, color=black, lw=1pt](C)(A)
			\Edge[Direct, color=black, lw=1pt](B)(C)

		\end{tikzpicture}
		\caption{Counterexample to the existence of an impartial $3$-selection mechanism that is $\alpha$-median-additive or $\alpha$-mean-additive on $\calG_3$ for $\alpha<1/2$. Vertices drawn in white have to be selected, vertices in black cannot be selected. 
			For the graphs at the top, on the left, and on the right, this follows from $\alpha$-median-additivity or $\alpha$-mean-additivity for $\alpha<1/2$.
			An arrow with label $v$ from one graph to another indicates that one can be obtained from the other by changing the outgoing edges of vertex~$v$; by impartiality, the vertex thus has to be selected in both graphs or not selected in both graphs. It follows that no vertices are selected in the graph at the center, a contradiction to the claimed additive guarantee.}
		\label{fig:lb_3}
	\end{figure}
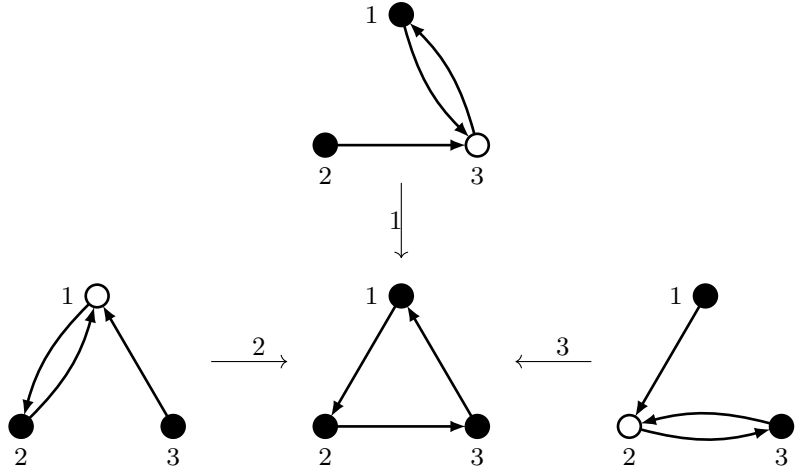

	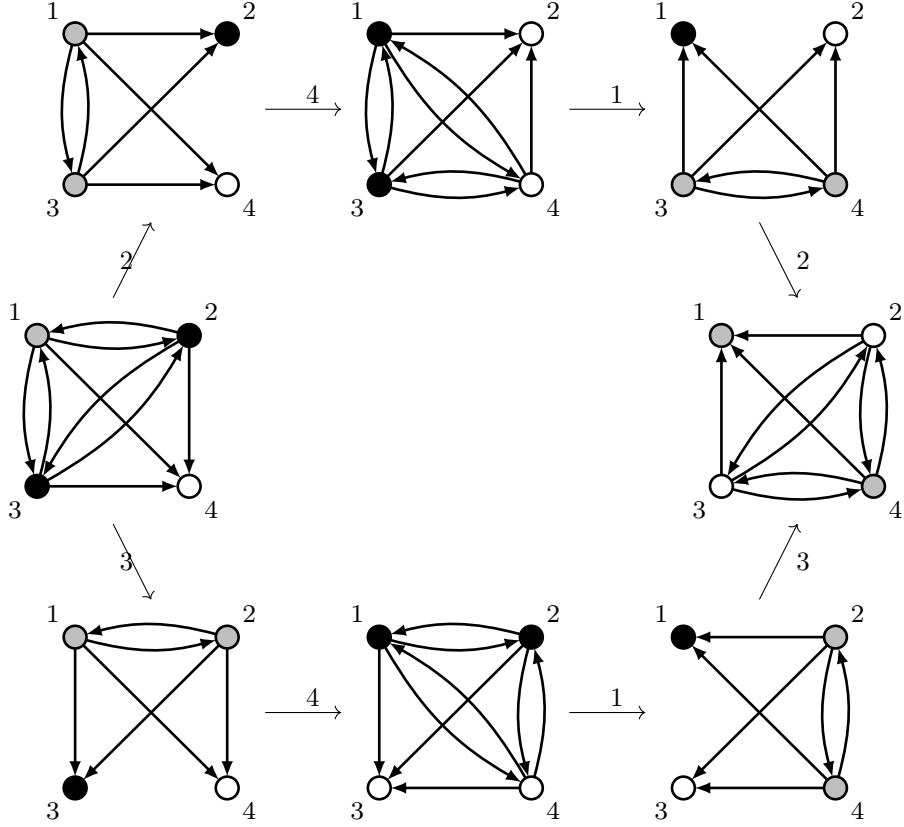
\begin{figure}[!ht]
		\centering
		\begin{tikzpicture}
			
			\Vertex[x=-.5,y=6, Math, shape=circle, color=lightgray, size=.3, label=1, fontscale=1.2, position=above left, distance=-.08cm]{A}
			\Vertex[x=1.5, y=6, Math, shape=circle, color=black, size=.3, label=2, fontscale=1.2, position=above right, distance=-.08cm]{B}
			\Vertex[x=-.5,y=4, Math, shape=circle, color=black, size=.3, label=3, fontscale=1.2, position=below left, distance=-.08cm]{C}
			\Vertex[x=1.5, y=4, Math, shape=circle, color=white, size=.3, label=4, fontscale=1.2, position=below right, distance=-.08cm]{D}
			\Edge[Direct, color=black, lw=1pt, bend=-15](A)(B)
			\Edge[Direct, color=black, lw=1pt, bend=-15](B)(A)
			\Edge[Direct, color=black, lw=1pt, bend=-15](A)(C)
			\Edge[Direct, color=black, lw=1pt, bend=-15](C)(A)
			\Edge[Direct, color=black, lw=1pt, bend=-15](B)(C)
			\Edge[Direct, color=black, lw=1pt, bend=-15](C)(B)
			\Edge[Direct, color=black, lw=1pt](A)(D)
			\Edge[Direct, color=black, lw=1pt](B)(D)
			\Edge[Direct, color=black, lw=1pt](C)(D)

			\draw[->] (.5,6.5) -- (1,7.5);
			\Text[x=.55, y=7, fontsize=\small]{$2$};
			
			\draw[->] (.5,3.5) -- (1,2.5);
			\Text[x=.55, y=3, fontsize=\small]{$3$};
			
			\Vertex[y=10, Math, shape=circle, color=lightgray, size=.3, label=1, fontscale=1.2, position=above left, distance=-.08cm]{A}
			\Vertex[x=2, y=10, Math, shape=circle, color=black, size=.3, label=2, fontscale=1.2, position=above right, distance=-.08cm]{B}
			\Vertex[y=8, Math, shape=circle, color=lightgray, size=.3, label=3, fontscale=1.2, position=below left, distance=-.08cm]{C}
			\Vertex[x=2, y=8, Math, shape=circle, color=white, size=.3, label=4, fontscale=1.2, position=below right, distance=-.08cm]{D}
			\Edge[Direct, color=black, lw=1pt](A)(B)
			\Edge[Direct, color=black, lw=1pt, bend=-15](A)(C)
			\Edge[Direct, color=black, lw=1pt, bend=-15](C)(A)
			\Edge[Direct, color=black, lw=1pt](C)(B)
			\Edge[Direct, color=black, lw=1pt](A)(D)
			\Edge[Direct, color=black, lw=1pt](C)(D)
			
			\draw[->] (2.5,9) -- (3.5,9);
			\Text[x=3, y=9.2, fontsize=\small]{$4$};
			
			\Vertex[y=2, Math, shape=circle, color=lightgray, size=.3, label=1, fontscale=1.2, position=above left, distance=-.08cm]{A}
			\Vertex[x=2, y=2, Math, shape=circle, color=lightgray, size=.3, label=2, fontscale=1.2, position=above right, distance=-.08cm]{B}
			\Vertex[Math, shape=circle, color=black, size=.3, label=3, fontscale=1.2, position=below left, distance=-.08cm]{C}
			\Vertex[x=2, Math, shape=circle, color=white, size=.3, label=4, fontscale=1.2, position=below right, distance=-.08cm]{D}
			\Edge[Direct, color=black, lw=1pt, bend=-15](A)(B)
			\Edge[Direct, color=black, lw=1pt, bend=-15](B)(A)
			\Edge[Direct, color=black, lw=1pt](A)(C)
			\Edge[Direct, color=black, lw=1pt](B)(C)
			\Edge[Direct, color=black, lw=1pt](A)(D)
			\Edge[Direct, color=black, lw=1pt](B)(D)
			
			\draw[->] (2.5,1) -- (3.5,1);
			\Text[x=3, y=1.2, fontsize=\small]{$4$};
			
			\Vertex[x=4, y=10, Math, shape=circle, color=black, size=.3, label=1, fontscale=1.2, position=above left, distance=-.08cm]{A}
			\Vertex[x=6, y=10, Math, shape=circle, color=white, size=.3, label=2, fontscale=1.2, position=above right, distance=-.08cm]{B}
			\Vertex[x=4, y=8, Math, shape=circle, color=black, size=.3, label=3, fontscale=1.2, position=below left, distance=-.08cm]{C}
			\Vertex[x=6, y=8, Math, shape=circle, color=white, size=.3, label=4, fontscale=1.2, position=below right, distance=-.08cm]{D}
			\Edge[Direct, color=black, lw=1pt](A)(B)
			\Edge[Direct, color=black, lw=1pt, bend=-15](A)(C)
			\Edge[Direct, color=black, lw=1pt, bend=-15](C)(A)
			\Edge[Direct, color=black, lw=1pt](C)(B)
			\Edge[Direct, color=black, lw=1pt, bend=-15](A)(D)
			\Edge[Direct, color=black, lw=1pt, bend=-15](C)(D)
			\Edge[Direct, color=black, lw=1pt, bend=-15](D)(A)
			\Edge[Direct, color=black, lw=1pt](D)(B)
			\Edge[Direct, color=black, lw=1pt, bend=-15](D)(C)
			
			\draw[->] (6.5,9) -- (7.5,9);
			\Text[x=7, y=9.2, fontsize=\small]{$1$};
			
			\Vertex[x=4, y=2, Math, shape=circle, color=black, size=.3, label=1, fontscale=1.2, position=above left, distance=-.08cm]{A}
			\Vertex[x=6, y=2, Math, shape=circle, color=black, size=.3, label=2, fontscale=1.2, position=above right, distance=-.08cm]{B}
			\Vertex[x=4, Math, shape=circle, color=white, size=.3, label=3, fontscale=1.2, position=below left, distance=-.08cm]{C}
			\Vertex[x=6, Math, shape=circle, color=white, size=.3, label=4, fontscale=1.2, position=below right, distance=-.08cm]{D}
			\Edge[Direct, color=black, lw=1pt, bend=-15](A)(B)
			\Edge[Direct, color=black, lw=1pt, bend=-15](B)(A)
			\Edge[Direct, color=black, lw=1pt](A)(C)
			\Edge[Direct, color=black, lw=1pt](B)(C)
			\Edge[Direct, color=black, lw=1pt, bend=-15](A)(D)
			\Edge[Direct, color=black, lw=1pt, bend=-15](B)(D)
			\Edge[Direct, color=black, lw=1pt, bend=-15](D)(A)
			\Edge[Direct, color=black, lw=1pt, bend=-15](D)(B)
			\Edge[Direct, color=black, lw=1pt](D)(C)
			
			\draw[->] (6.5,1) -- (7.5,1);
			\Text[x=7, y=1.2, fontsize=\small]{$1$};
			
			\Vertex[x=8, y=10, Math, shape=circle, color=black, size=.3, label=1, fontscale=1.2, position=above left, distance=-.08cm]{A}
			\Vertex[x=10, y=10, Math, shape=circle, color=white, size=.3, label=2, fontscale=1.2, position=above right, distance=-.08cm]{B}
			\Vertex[x=8, y=8, Math, shape=circle, color=lightgray, size=.3, label=3, fontscale=1.2, position=below left, distance=-.08cm]{C}
			\Vertex[x=10, y=8, Math, shape=circle, color=lightgray, size=.3, label=4, fontscale=1.2, position=below right, distance=-.08cm]{D}
			\Edge[Direct, color=black, lw=1pt](C)(A)
			\Edge[Direct, color=black, lw=1pt](C)(B)
			\Edge[Direct, color=black, lw=1pt, bend=-15](C)(D)
			\Edge[Direct, color=black, lw=1pt](D)(A)
			\Edge[Direct, color=black, lw=1pt](D)(B)
			\Edge[Direct, color=black, lw=1pt, bend=-15](D)(C)
			
			\Vertex[x=8, y=2, Math, shape=circle, color=black, size=.3, label=1, fontscale=1.2, position=above left, distance=-.08cm]{A}
			\Vertex[x=10, y=2, Math, shape=circle, color=lightgray, size=.3, label=2, fontscale=1.2, position=above right, distance=-.08cm]{B}
			\Vertex[x=8, Math, shape=circle, color=white, size=.3, label=3, fontscale=1.2, position=below left, distance=-.08cm]{C}
			\Vertex[x=10, Math, shape=circle, color=lightgray, size=.3, label=4, fontscale=1.2, position=below right, distance=-.08cm]{D}
			\Edge[Direct, color=black, lw=1pt](B)(A)
			\Edge[Direct, color=black, lw=1pt](B)(C)
			\Edge[Direct, color=black, lw=1pt, bend=-15](B)(D)
			\Edge[Direct, color=black, lw=1pt](D)(A)
			\Edge[Direct, color=black, lw=1pt, bend=-15](D)(B)
			\Edge[Direct, color=black, lw=1pt](D)(C)
			
			\draw[->] (9,7.5) -- (9.5,6.5);
			\Text[x=9.45, y=7, fontsize=\small]{$2$};
			
			\draw[->] (9,2.5) -- (9.5,3.5);
			\Text[x=9.45, y=3, fontsize=\small]{$3$};
			
			\Vertex[x=8.5, y=6, Math, shape=circle, color=lightgray, size=.3, label=1, fontscale=1.2, position=above left, distance=-.08cm]{A}
			\Vertex[x=10.5, y=6, Math, shape=circle, color=white, size=.3, label=2, fontscale=1.2, position=above right, distance=-.08cm]{B}
			\Vertex[x=8.5, y=4, Math, shape=circle, color=white, size=.3, label=3, fontscale=1.2, position=below left, distance=-.08cm]{C}
			\Vertex[x=10.5, y=4, Math, shape=circle, color=lightgray, size=.3, label=4, fontscale=1.2, position=below right, distance=-.08cm]{D}
			\Edge[Direct, color=black, lw=1pt](B)(A)
			\Edge[Direct, color=black, lw=1pt, bend=-15](B)(C)
			\Edge[Direct, color=black, lw=1pt, bend=-15](B)(D)
			\Edge[Direct, color=black, lw=1pt](C)(A)
			\Edge[Direct, color=black, lw=1pt, bend=-15](C)(B)
			\Edge[Direct, color=black, lw=1pt, bend=-15](C)(D)
			\Edge[Direct, color=black, lw=1pt](D)(A)
			\Edge[Direct, color=black, lw=1pt, bend=-15](D)(B)
			\Edge[Direct, color=black, lw=1pt, bend=-15](D)(C)
			
		\end{tikzpicture}
		\caption{Counterexample to the existence of an impartial $4$-selection mechanism that is $\alpha_1$-median-additive on $\calG_4(3)$ for $\alpha_1<1$ or $\alpha_2$-mean-additive on $\calG_4$ for $\alpha_2<2/3$. Vertices drawn in white have to be selected, vertices in black cannot be selected, and vertices in gray may or may not be selected. 
			For the graph on the left, this follows from $\alpha_1$-median-additivity for $\alpha_1<1$ or $\alpha_2$-mean-additivity for $\alpha_2<2/3$: under these assumptions at most one of the vertices with indegree 2 can be selected, which without loss of generality we can assume to be vertex~$1$. For the other graphs, it then follows by impartiality, and for the graph on the right yields a contradiction to the claimed additive guarantees.}
		\label{fig:lb_4}
	\end{figure}
	
	The median of any set of numbers is an upper bound on their minimum. Therefore, if no impartial mechanism exists that is $\alpha$-median-additive on $\calG'\subseteq\calG$ for $\alpha<\bar{\alpha}$, then no impartial mechanism can exist that is $\alpha$-min-additive on $\calG'$ for $\alpha<\lceil \bar{\alpha}\rceil$. We thus obtain the following impossibility result, which we have claimed previously.
	\begin{corollary}
		Let $n\in\NN,\ n\geq 3$, and $k\in [n]$. Let $f$ be an $\alpha$-min-additive impartial $k$-selection mechanism on $\calG_n$. Then $\alpha\geq 1$.
	\end{corollary}
	
	The impossibility results imply that for $k\geq d+1$, the mechanisms of \autoref{sec:pwru} are best possible for the minimum and median objectives except in a few boundary cases. When $n=2$, selecting each of the two vertices if and only if it has an incoming edge is impartial and achieves $0$-min-additivity and $0$-median-additivity. When $n=3$, it is possible to select in an impartial way at least one vertex with maximum indegree and at most one vertex with indegree equal to the maximum indegree minus one, thus guaranteeing $1/2$-median-additivity.
	For the mean objective, the mechanisms of \autoref{sec:pwru} are best possible asymptotically under the additional assumption that $k=O(d)$.
	
	It is worth pointing out that the proof of the impossibility result uses graphs in which some vertices, in particular those with maximum indegree, do not have any outgoing edges. However, the impossibility extends naturally to the case where this cannot happen, corresponding to the practically relevant case in which abstentions are not allowed, as long as $n\geq 4$ and $d\geq 3$. For this it is enough to define $D=[d]$, add a new vertex with outgoing edges to every vertex in $D$ and incoming edges from the vertices in $D$ which do not have any outgoing edge, and construct a cycle containing the vertices in $V\setminus D$.
	
	\section{Trading Off Quantity and Quality}
	\label{sec:edge-deletion}
	
	We have so far given impartial selection mechanisms for settings where the maximum outdegree~$d$ is smaller than the maximum number~$k$ of vertices that can be selected,
	and have shown that the mechanisms provide best possible additive guarantees in such settings. We will now consider settings where~$d\geq k$, such that asymmetric plurality with runners-up selects too many vertices and therefore cannot be used directly. 
	We obtain the following result.
	\begin{theorem}
		\label{thm:ub-deletion}
		For every $n\in \NN$ and $k\in \{2,\ldots,n\}$, there exists an impartial and $(\lfloor (n-2)/(k-1)\rfloor+1)$-min-additive $k$-selection mechanism on $\calG_n$.
	\end{theorem}
	
	The result is obtained by a variant of asymmetric plurality with runners-up in which some edges are deleted before the mechanism is run.
	In principle, deleting a certain number of edges can affect the additive guarantee by the same amount, if all of the deleted edges happen to be directed at the same vertex.
	By studying the structure of the set of vertices selected by the mechanism, we will instead be able to delete edges to distinct vertices and thus keep the negative impact on the additive guarantee under control.
	
	The modified mechanism, which we call \textit{asymmetric plurality with runners-up and edge deletion}, is formally described in \autoref{alg:PWRU-deletion}. It deletes any edges from a vertex to the $\lfloor(n-2)/(k-1)\rfloor$ vertices preceding that vertex in the tie-breaking order, and applies asymmetric plurality with runners-up to the resulting graph. The following lemma shows that without such edges, the maximum number of vertices selected is reduced to~$k$.
	
	\begin{algorithm}[t]
		\SetAlgoNoLine
		\KwIn{Digraph $G=(V,E)\in \calG_n$, $k\in \{2,\ldots,n\}$}
		\KwOut{Set $S\subseteq V$ of selected vertices with $|S|\leq k$}
		Let $r=\left\lfloor (n-2)/(k-1)\right\rfloor$ \tcp*{number of outgoing edges to remove}
		Let $R=\bigcup_{u=1}^{n-1}\bigcup_{v=u+1}^{\min\{u+r,n\}}\{(u,v)\}$ \tcp*{edges to be removed}
		Let $\bar{G}=(V, E\setminus R)$\;
		{\bf Return} $\calP(\bar{G})$
		\caption{Asymmetric plurality with runners-up and edge deletion $\mathsf{\calP^D(G)}$}
		\label{alg:PWRU-deletion}
	\end{algorithm}
	
	\begin{lemma}\label{lem:deletion}
		Let $n\in \NN,\ k\in \{2,\ldots,n\}$, and $r\in \NN$ with $r\geq \lfloor  (n-2)/(k-1)\rfloor$. Let $G=(V,E)\in \calG_n$ be such that for every $u\in \{1,\ldots,n-1\}$ and every $v\in \{u+1,\ldots, \min\{u+r,n\}\},\ (u,v)\notin E$.
		Then, $|\calP(G)|\leq k$.
	\end{lemma}
	
	\begin{proof}
		As in the proof of \autoref{thm:ub-1-n-1}, we let $S^i=\{v\in \calP(G):\delta^-(v)=\Delta(G)-i\}$ and $n^i=|S^i|$ for $i\in \{0,1\}$, and now we denote its elements in increasing order by $v^i_j$ for $j\in [n^i]$, \ie
		\[
		S^i=\{v^i_j\}_{j=1}^{n^i} \text{ with } v^i_1<v^i_2\dots <v^i_{n^i} \text{ for each } i\in \{0,1\}.
		\]
		From \autoref{lem:char-pwru}, we know that $\calP(G)= S^0\cup S^1$, that for $i\in\{0,1\}$ we have $(v^i_j,v^i_k)\in E$ for every $j,k$ with $j<k$, and that $v^1_{1}>v^0_{n^0}$. This allows to define, for $i\in \{0,1\}$,
		\[
		\bar{S}^i=\{v\in V\setminus S_i: v^i_1<v<v^i_{n^i}\},\quad \bar{n}^i=|\bar{S}^i|,
		\]
		such that $\bar{S}^0\cap \bar{S}^1=\emptyset$.
		
		Fix $i\in\{0,1\}$ and suppose that $n^i\geq 2$. Combining both the fact that $(v^i_j,v^i_k)\in E$ for every $j, k$ with $j<k$, and that for every $u\in \{1,\ldots,n-1\}$ and $v\in \{u+1,\ldots,\min\{u+r,n\}\},\ (u,v)\notin E$, we have that for every $j\in [n^i-1]$ it holds $v^i_{j+1}-v^i_{j}\geq r+1$. Summing over $j$ yields $v^i_{n^i}-v^i_1\geq (n^i-1)(r+1)$, hence
		\[
		\bar{n}^i=v^i_{n^i}-v^i_1+1-n^i \geq (n^i-1)(r+1)+1-n^i =(n^i-1)r,
		\]
		where the first equality comes from the definition of the set $\bar{S}^i$. 
		This implies $n^i\leq 1+ \bar{n}^i/r$. 
		We can now lift the assumption $n^i\geq 2$, since when $n^i=1$ we have $\bar{n}^i=0$ and the inequality holds as well, and write the following chain of inequalities:
		\[
		|\calP(G)|=n^0+n^1 \leq 2+\frac{\bar{n}^0+\bar{n}^1}{r} \leq 2+\frac{n-|\calP(G)|}{r},
		\]
		where the last inequality comes from the fact that all the sets $S^0,\ S^1,\ \bar{S}^0,\ \bar{S}^1$ are disjoint and therefore their cardinalities sum up to at most $n$. This bounds the number of selected vertices as $|\calP(G)|\leq (2r+n)/(r+1)$.
		
		Suppose now that $|\calP(G)|\geq k+1$. Using the previous bound, this yields
		\[
		2r+n  \geq (k+1)(r+1) \Longleftrightarrow r \leq \frac{n-k-1}{k-1} = \frac{n-2}{k-1}-1,
		\]
		which contradicts the lower bound on $r$ in the statement of the lemma.
	\end{proof}
	
	\autoref{fig:example-pwru-deletion} illustrates the argument and notation of \autoref{lem:deletion}.
	We are now ready to prove \autoref{thm:ub-deletion}.
	
	\begin{figure}[t]
		\centering
		\begin{tikzpicture}
			
			\Vertex[x=5.5, y=3, Math, shape=circle, color=white, size=.1, label=v^0_{n^0}, fontscale=1.2, position=above, distance=-.04cm]{A}
			\Text[x=7.25, y=3]{$\dots$}
			\Vertex[x=9, y=3, Math, shape=circle, color=white, size=.1, label=v^0_{2}, fontscale=1.2, position=above, distance=-.04cm]{B}
			\Vertex[x=10.3, y=3, Math, shape=circle, color=white, size=.1, label=v^0_{1}, fontscale=1.2, position=above, distance=-.04cm]{C}
			
			\Vertex[y=1.5, Math, shape=circle, color=white, size=.1, label=v^1_{n^1}, fontscale=1.2, position=below, distance=-.04cm]{D}
			\Text[x=1.75, y=1.5]{$\dots$}
			\Vertex[x=3.5, y=1.5, Math, shape=circle, color=white, size=.1, label=v^1_{2}, fontscale=1.2, position=below, distance=-.04cm]{E}
			\Vertex[x=4.8, y=1.5, Math, shape=circle, color=white, size=.1, label=v^1_{1}, fontscale=1.2, position=below, distance=-.04cm]{F}
			
			\Edge[Direct, color=black, lw=1pt, bend=-20](B)(A)
			\Edge[Direct, color=black, lw=1pt, bend=15](C)(A)
			\Edge[Direct, color=black, lw=1pt](C)(B)
			\Edge[Direct, color=black, lw=1pt, bend=20](E)(D)
			\Edge[Direct, color=black, lw=1pt, bend=-15](F)(D)
			\Edge[Direct, color=black, lw=1pt](F)(E)
			
			\Text[x=-1, y=3, fontsize=\small]{$\Delta$};
			\Text[x=-1, y=1.5, fontsize=\small]{$\Delta-1$};
			
			\Vertex[x=.3, Math, shape=circle, color=black, size=.1]{G}
			\Text[x=.68]{$\dots$}
			\Vertex[x=1, Math, shape=circle, color=black, size=.1]{H}
			\Text[x=.65, y=-.35]{$\underbrace{\hspace{1cm}}_{\geq r}$}
			
			\Text[x=1.75]{$\dots$}
			
			\Vertex[x=2.5, Math, shape=circle, color=black, size=.1]{I}
			\Text[x=2.88]{$\dots$}
			\Vertex[x=3.2, Math, shape=circle, color=black, size=.1]{J}
			\Text[x=2.85, y=-.35]{$\underbrace{\hspace{1cm}}_{\geq r}$}
			
			\Vertex[x=3.8, Math, shape=circle, color=black, size=.1]{K}
			\Text[x=4.18]{$\dots$}
			\Vertex[x=4.5, Math, shape=circle, color=black, size=.1]{L}
			\Text[x=4.15, y=-.35]{$\underbrace{\hspace{1cm}}_{\geq r}$}

			\Vertex[x=5.8, Math, shape=circle, color=black, size=.1]{G}
			\Text[x=6.18]{$\dots$}
			\Vertex[x=6.5, Math, shape=circle, color=black, size=.1]{H}
			\Text[x=6.15, y=-.35]{$\underbrace{\hspace{1cm}}_{\geq r}$}
			
			\Text[x=7.25]{$\dots$}
			
			\Vertex[x=8, Math, shape=circle, color=black, size=.1]{I}
			\Text[x=8.38]{$\dots$}
			\Vertex[x=8.7, Math, shape=circle, color=black, size=.1]{J}
			\Text[x=8.35, y=-.35]{$\underbrace{\hspace{1cm}}_{\geq r}$}
			
			\Vertex[x=9.3, Math, shape=circle, color=black, size=.1]{K}
			\Text[x=9.68]{$\dots$}
			\Vertex[x=10, Math, shape=circle, color=black, size=.1]{L}
			\Text[x=9.65, y=-.35]{$\underbrace{\hspace{1cm}}_{\geq r}$}
			
			\draw[dashed] (-.4,.8) rectangle (5.2,2.1);
			\Text[x=2.4, y=2.4]{$S_1$};
			
			\draw[dashed] (0,-.75) rectangle (4.8,.3);
			\Text[x=2.4, y=-1.05]{$\bar{S}_1$};
			
			\draw[dashed] (5.1,2.4) rectangle (10.7,3.7);
			\Text[x=7.9, y=2.1]{$S_0$};
			
			\draw[dashed] (5.5,-.75) rectangle (10.3,.3);
			\Text[x=7.9, y=-1.05]{$\bar{S}_0$};
			
		\end{tikzpicture}
		\caption{Illustration of \autoref{lem:deletion}. There are no edges from a vertex to any of the $r$ vertices to its left, which means that for each vertex in $S_0$ or $S_1$, except for the left-most vertex, there exist are at least $r$ vertices outside these sets. Such vertices are not arranged according to their indegrees, and edges from vertices in $S_1$ to every vertex in $S_0$ have been omitted for clarity.}
		\label{fig:example-pwru-deletion}
	\end{figure}
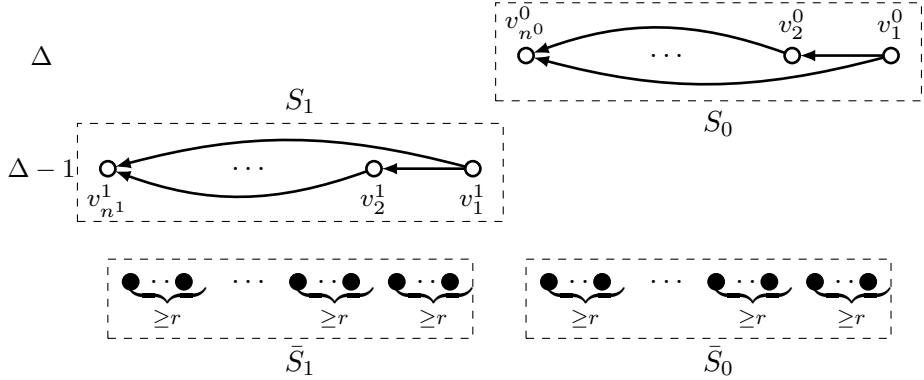
	
	\begin{proof}[Proof of \autoref{thm:ub-deletion}]
		We show that \autoref{alg:PWRU-deletion} satisfies the conditions of the theorem. Let $n\in \NN$ and $k\in \{2,\ldots,n\}$.
		Impartiality follows from the fact that \autoref{alg:PWRU-order} is impartial, thus the potential deletion of outgoing edges of a given vertex cannot affect the fact of selecting this vertex or not.
		Formally, if $G=(V,E),\ v\in V$ and $G'=(V,E')\in \calN_v(G)$, then defining $\bar{G}=(V,\bar{E})$ and $\bar{G}'=(V,\bar{E}')$ as the graphs constructed when running \autoref{alg:PWRU-deletion} with $G$ and $G'$ as input graphs, respectively, we have 
		\begin{align*}
			\bar{E}\setminus (\{v\}\times V) & = (E\setminus  (\{v\}\times V))\setminus \left(\bigcup_{u=1}^{n-1}\bigcup_{w=u+1}^{\min\{u+r,n\}}\{(u,w)\}\right)\\
			& = (E'\setminus  (\{v\}\times V))\setminus \left(\bigcup_{u=1}^{n-1}\bigcup_{w=u+1}^{\min\{u+r,n\}}\{(u,w)\}\right)\\
			& = \bar{E}'\setminus (\{v\}\times V),
		\end{align*}
		where we use that $G'\in \calN_v(G)$.
		Impartiality then follows directly from impartiality of plurality with runners-up.
		For the following, let $G=(V,E)\in \calG_n$ and define $r$ and $\bar{G}$ as in the mechanism.
		Since the first step of the mechanism ensures that for every $u\in \{1,\ldots,n-1\}$ and every $v\in \{u+1,\ldots, \min\{u+r,n\}\},\ (u,v)\notin E$, \autoref{lem:deletion} implies that $|\calP^D(G)|=|\calP(\bar{G})|\leq k$.
		Finally, in order to show the additive guarantee we first note that, for every $v\in V,\ \delta^-(v, G) \leq \delta^-(v, \bar{G})+r$, since at most $|\{v-r,\ldots,v-1\}\cap V|\leq r$ incoming edges of $v$ are deleted when defining $\bar{G}$ from $G$.
		In particular, $\Delta(G) \leq \Delta(\bar{G})+r$.
		Using this observation and denoting $v^*\in\text{argmin}_{v\in \calP^D(G)}\{\delta^-(v,G)\}$ an arbitrary element with minimum indegree among those selected by asymmetric plurality with runners-up and edge deletion, we obtain that
		\[
		\delta^-(v^*,G) \geq \delta^-(v^*,\bar{G}) \geq \Delta(\bar{G})-1 \geq \Delta(G)-r - 1,
		\]
		where the second inequality comes from \autoref{lem:char-pwru}, since $v^*$ belongs to $\calP(\bar{G})$.
		We conclude that the mechanism is $(r+1)$-min-additive for $r=\lfloor (n-2)/(k-1)\rfloor$.
	\end{proof}

	It is easy to see that the previous analysis is tight from a graph $G=(V,E)$ where exactly $r=\lfloor (n-2)/(k-1)\rfloor$ incoming edges of the top-voted vertex are deleted, and a vertex with the second highest indegree $u$ such that $u>\text{top}(G),\ (u,\text{top}(G))\in E$, and $\delta^-(u)=\Delta(G)-r-1$ is selected. However, we do not know whether the tradeoff provided by \autoref{thm:ub-deletion} is best possible for any impartial mechanism, and the question for the optimum tradeoff is an interesting one. Currently, when $d\geq k$ a gap remains between the upper bound of $\lfloor (n-2)/(k-1)\rfloor+1$ and a lower bound of $1$, which is relatively large when the number $k$ of vertices that can be selected is small. We may, alternatively, also ask for the number of vertices that have to be selected in order to guarantee $1$-min-additivity. Currently, the best upper bound on this number is $n-1$.
	
	In addition to the question about the performance of the mechanism introduced in this section, the sole fact that sometimes it does not select vertices with indegree strictly higher than the one of other selected vertices may seem unfair.
	Unfortunately, this is unavoidable whenever $d\geq k$ and $\alpha$-min-additivity is imposed for some $\alpha<d$, as one can see from a graph consisting of a complete subgraph on $d+1$ vertices and $n-(d+1)$ isolated vertices. For any $k$-selection mechanism, a vertex in the complete subgraph is not selected, and impartiality forces us to not select it either when its outgoing edges are deleted and it is the unique top-voted vertex.

	\subsubsection{Acknowledgments}
	
	The authors have benefitted from discussions with David Hannon. Research was supported by the Deutsche Forschungsgemeinschaft under project number~431465007 and by the Engineering and Physical Sciences Research Council under grant~EP/T015187/1.
	
	\bibliographystyle{abbrvnat}
	\bibliography{WINE-bibliography}

\begin{thebibliography}{21}
\providecommand{\natexlab}[1]{#1}
\providecommand{\url}[1]{\texttt{#1}}
\expandafter\ifx\csname urlstyle\endcsname\relax
  \providecommand{\doi}[1]{doi: #1}\else
  \providecommand{\doi}{doi: \begingroup \urlstyle{rm}\Url}\fi

\bibitem[Alon et~al.(2011)Alon, Fischer, Procaccia, and
  Tennenholtz]{alon2011sum}
N.~Alon, F.~Fischer, A.~Procaccia, and M.~Tennenholtz.
\newblock Sum of us: Strategyproof selection from the selectors.
\newblock In \emph{Proceedings of the 13th Conference on Theoretical Aspects of
  Rationality and Knowledge}, pages 101--110, 2011.

\bibitem[Aziz et~al.(2019)Aziz, Lev, Mattei, Rosenschein, and
  Walsh]{aziz2019strategyproof}
H.~Aziz, O.~Lev, N.~Mattei, J.~S. Rosenschein, and T.~Walsh.
\newblock Strategyproof peer selection using randomization, partitioning, and
  apportionment.
\newblock \emph{Artificial Intelligence}, 275:\penalty0 295--309, 2019.

\bibitem[Babichenko et~al.(2020)Babichenko, Dean, and
  Tennenholtz]{babichenko2020incentive}
Y.~Babichenko, O.~Dean, and M.~Tennenholtz.
\newblock Incentive-compatible selection mechanisms for forests.
\newblock In \emph{Proceedings of the 21st ACM Conference on Economics and
  Computation}, pages 111--131, 2020.

\bibitem[Bjelde et~al.(2017)Bjelde, Fischer, and Klimm]{bjelde2017impartial}
A.~Bjelde, F.~Fischer, and M.~Klimm.
\newblock Impartial selection and the power of up to two choices.
\newblock \emph{ACM Transactions on Economics and Computation}, 5\penalty0
  (4):\penalty0 1--20, 2017.

\bibitem[Bousquet et~al.(2014)Bousquet, Norin, and Vetta]{bousquet2014near}
N.~Bousquet, S.~Norin, and A.~Vetta.
\newblock A near-optimal mechanism for impartial selection.
\newblock In \emph{Proceedings of the 10th International Conference on Web and
  Internet Economics}, pages 133--146. Springer, 2014.

\bibitem[Caragiannis et~al.(2019)Caragiannis, Christodoulou, and
  Protopapas]{caragiannis2019impartial}
I.~Caragiannis, G.~Christodoulou, and N.~Protopapas.
\newblock Impartial selection with additive approximation guarantees.
\newblock In \emph{Proceedings of the 12th International Symposium on
  Algorithmic Game Theory}, pages 269--283. Springer, 2019.

\bibitem[Caragiannis et~al.(2021)Caragiannis, Christodoulou, and
  Protopapas]{caragiannis2021impartial}
I.~Caragiannis, G.~Christodoulou, and N.~Protopapas.
\newblock Impartial selection with prior information.
\newblock \emph{arXiv preprint arXiv:2102.09002}, 2021.

\bibitem[Cembrano et~al.(2022)Cembrano, Fischer, Hannon, and
  Klimm]{cembrano2022impartial}
J.~Cembrano, F.~Fischer, D.~Hannon, and M.~Klimm.
\newblock Impartial selection with additive guarantees via iterated deletion.
\newblock \emph{arXiv preprint arXiv:2205.08979}, 2022.

\bibitem[de~Clippel et~al.(2008)de~Clippel, Moulin, and
  Tideman]{de2008impartial}
G.~de~Clippel, H.~Moulin, and N.~Tideman.
\newblock Impartial division of a dollar.
\newblock \emph{Journal of Economic Theory}, 139\penalty0 (1):\penalty0
  176--191, 2008.

\bibitem[Fischer and Klimm(2015)]{fischer2015optimal}
F.~Fischer and M.~Klimm.
\newblock Optimal impartial selection.
\newblock \emph{SIAM Journal on Computing}, 44\penalty0 (5):\penalty0
  1263--1285, 2015.

\bibitem[Holzman and Moulin(2013)]{holzman2013impartial}
R.~Holzman and H.~Moulin.
\newblock Impartial nominations for a prize.
\newblock \emph{Econometrica}, 81\penalty0 (1):\penalty0 173--196, 2013.

\bibitem[Kahng et~al.(2018)Kahng, Kotturi, Kulkarni, Kurokawa, and
  Procaccia]{kahng2018ranking}
A.~Kahng, Y.~Kotturi, C.~Kulkarni, D.~Kurokawa, and A.~D. Procaccia.
\newblock Ranking wily people who rank each other.
\newblock In \emph{Proceedings of the 32nd AAAI Conference on Artificial
  Intelligence}, 2018.

\bibitem[Kurokawa et~al.(2015)Kurokawa, Lev, Morgenstern, and
  Procaccia]{kurokawa2015impartial}
D.~Kurokawa, O.~Lev, J.~Morgenstern, and A.~D. Procaccia.
\newblock Impartial peer review.
\newblock In \emph{Proceedings of the 24th International Joint Conference on
  Artificial Intelligence}, 2015.

\bibitem[Mackenzie(2015)]{mackenzie2015symmetry}
A.~Mackenzie.
\newblock Symmetry and impartial lotteries.
\newblock \emph{Games and Economic Behavior}, 94:\penalty0 15--28, 2015.

\bibitem[Mackenzie(2020)]{MacK20a}
A.~Mackenzie.
\newblock An axiomatic analysis of the papal conclave.
\newblock \emph{Economic Theory}, 69:\penalty0 713--743, 2020.

\bibitem[Mattei et~al.(2020)Mattei, Turrini, and
  Zhydkov]{mattei2020peernomination}
N.~Mattei, P.~Turrini, and S.~Zhydkov.
\newblock Peernomination: Relaxing exactness for increased accuracy in peer
  selection.
\newblock \emph{arXiv preprint arXiv:2004.14939}, 2020.

\bibitem[Tamura(2016)]{Tamu16a}
S.~Tamura.
\newblock Characterizing minimal impartial rules for awarding prizes.
\newblock \emph{Games and Economic Behavior}, 95:\penalty0 41--46, 2016.

\bibitem[Tamura and Ohseto(2014)]{TaOh14a}
S.~Tamura and S.~Ohseto.
\newblock Impartial nomination correspondences.
\newblock \emph{Social Choice and Welfare}, 43\penalty0 (1):\penalty0 47--54,
  2014.

\bibitem[W{\k{a}}s et~al.(2019)W{\k{a}}s, Rahwan, and Skibski]{wkas2019random}
T.~W{\k{a}}s, T.~Rahwan, and O.~Skibski.
\newblock Random walk decay centrality.
\newblock In \emph{Proceedings of the AAAI Conference on Artificial
  Intelligence}, volume~33, pages 2197--2204, 2019.

\bibitem[Xu et~al.(2018)Xu, Zhao, Shi, Zhang, and Shah]{xu2018strategyproof}
Y.~Xu, H.~Zhao, X.~Shi, J.~Zhang, and N.~B. Shah.
\newblock On strategyproof conference peer review.
\newblock \emph{arXiv preprint arXiv:1806.06266}, 2018.

\bibitem[Zhang et~al.(2021)Zhang, Zhang, and Zhao]{zhang2021incentive}
X.~Zhang, Y.~Zhang, and D.~Zhao.
\newblock Incentive compatible mechanism for influential agent selection.
\newblock In \emph{Proceedings of the 14th International Symposium on
  Algorithmic Game Theory}, pages 79--93. Springer, 2021.

\end{thebibliography}
	
\end{document}